\def\IEEEsubmission{0}
\def\isConference{0}

\if\isConference1
	\documentclass[conference]{IEEEtran}
\else
	\if\IEEEsubmission1
	\documentclass[journal,12pt,onecolumn,draftclsnofoot]{IEEEtran}
	\else
	\documentclass[journal]{IEEEtran}
	\fi
\fi

\usepackage{makecell}
\usepackage{mathtools}
\usepackage{amsmath,amsfonts,amssymb,amsthm}
\usepackage{algorithmic}
\usepackage{algorithm}
\usepackage{array}
\usepackage{acronym}
\usepackage[caption=false,font=footnotesize]{subfig}
\usepackage{textcomp}
\usepackage{stfloats}
\usepackage{url}
\usepackage{verbatim}
\usepackage{color}
\usepackage{graphicx}
\usepackage{cite}
\def\figuresizeBB{5.2in}

\def\figuresize{2.8in}

\def\figuresizeSS{1.7in}
\def\figuresizeSSS{1.13in}
\def\complexNumbers{\mathbb{C}}

\def\functionSpace[#1]{\mathcal{F}(#1)}
\def\realNumbers{\mathbb{R}}
\def\integers{\mathbb{Z}}
\def\constante{{\rm e}}
\def\constanti{{\rm i}}
\def\expectationOperator[#1][#2]{\mathbb{E}_{#2}[#1]}
\def\varianceOperator[#1]{\mathrm{Var}\left(#1\right)}
\def\uniformDistribution[#1][#2]{{\mathcal{U}_{[#1,#2]}}}
\def\traceOperator[#1]{{\mathrm{tr}}\{#1\}}
\def\identityMatrix[#1]{\mathrm{\textbf{I}}_{#1}}
\def\zeroVector[#1]{{ {{\mathbf{0}}}}_{#1}}
\def\oneVector[#1]{{ {\mathbf{1}}}_{#1}}
\def\exponentialIntegral[#1]{\mathrm{Ei}(#1)}

\def\indicatorFunction[#1]{\mathbb{I}[{#1}]}
\def\PDFimpairments[#1]{f_{\rm impairment}(#1)}

\def\clamp[#1][#2]{\text{clamp}_{#2}\left(#1\right)}
\def\signNormal[#1]{\text{sign}\left(#1\right)}
\def\diagOperation[#1]{\text{diag}\left\{#1\right\}}
\def\arithmeticMean[#1]{\overline{#1}}

\def\qfunction[#1]{\mathrm{Q}\left({#1}\right)}
\def\probability[#1]{\mathrm{Pr}\left({#1}\right)}
\def\complexGaussian[#1][#2]{\mathcal{CN}({#1,#2})}
\def\gaussian[#1][#2]{\mathcal{N}({#1,#2})}
\def\rayleigh[#1]{\mathrm{Rayleigh}({#1})}

\def\normalPDF[#1]{\phi\left(#1\right)}
\def\normalCDF[#1]{\Phi\left(#1\right)}

\def\CDF[#1][#2][#3]{F_{#1}^{#2}\left({#3}\right)}
\def\PDF[#1][#2][#3]{f_{#1}^{#2}\left({#3}\right)}

\def\complexGaussian[#1][#2]{\mathcal{CN}({#1,#2})}

\def\indexED{k}

\def\numberOfDigits{D}
\def\indexDigit{d}
\def\indexOACsymbolsElement{n}
\def\basis{{\boldsymbol\beta}}
\def\basisOpt{\hat{{\boldsymbol\beta}}}
\def\basisElement[#1]{\beta_{#1}}

\def\setOfParameters{\mathcal{X}}
\def\parameter[#1]{a_{#1}}

\def\indexParameters{q}
\def\numberOfParameters{Q}

\def\functionRange{\mathcal{Y}}
\def\fcnOutput[#1]{b_{#1}}
\def\indexFcnOutput{p}
\def\numberOfFcnOutputs{P}

\def\setOfSymbols{\mathcal{Q}}
\def\symbol[#1]{s_{#1}}
\def\symbolDetected[#1]{\hat{s}_{#1}}

\def\symbolSequence{\textit{\textbf{s}}}

\def\setOfCategories{\mathcal{P}}
\def\category{z}
\def\categoryDetected{\hat{z}}

\def\setOfHistogramsSub{\mathcal{T}_{K,Q,\numberOfDigits}}

\def\setOfHistograms[#1][#2]{\mathcal{T}_{#1,#2}}
\def\indexHistogram{r}
\def\histogram[#1]{{\textbf{t}_{#1}}}
\def\histogramDetected[#1]{{\hat{\textbf{t}}_{#1}}}

\def\mappingParameterToSymbol{\mathcal{M}_{\rm x}}

\def\mappingOutputToCategory{\mathcal{M}_{\rm y}}

\def\functionArgumentSequence{\textit{\textbf{x}}}
\def\functionArgument[#1]{x_{#1}}
\def\functionOutput{y}
\def\function{f}
\def\functionArbitrary[#1][#2]{\function_{#1}(#2)}

\def\functionSymbol{\tilde{f}}
\def\functionSymbolArbitrary[#1][#2]{\functionSymbol{#1}(#2)}

\def\functionHistogram{\tilde{f}_{\rm h}}
\def\functionHistogramArbitrary[#1]{\functionHistogram\left(#1\right)}

\def\functionDomain{\setOfParameters^\numberOfEdgeDevices}

\def\encoder{\epsilon}
\def\decoder{\delta}

\def\maxMultiplicty{R}

\def\functionOutputDetector{\hat{y}}

\def\numberOfEdgeDevices{K}
\def\numberOfEdgeDevicesForAGivenSymbol[#1]{K_{#1}}

\def\typeVector[#1]{\textbf{t}_{#1}}

\def\noiseVector{\textbf{n}}
\def\noiseVectorEle[#1]{\textbf{n}}

\def\setOfOACsymbols{\mathcal{C}}
\def\OACsymbolsMatrix{\textbf{C}}
\def\OACsymbolsMatrixBasis[#1]{\textbf{C}_{#1}}
\def\OACsymbolsMatrixOpt{\hat{\textbf{C}}}
\def\OACsymbolsMatrixUnnormalized{\hat{\textbf{C}}}
\def\OACsymbolUnnormalized[#1]{\hat{\textbf{c}}_{#1}}
\def\OACsymbolUnnormalizedElement[#1]{\hat{c}_{#1}}

\def\OACsymbolMean{\mu}
\def\OACsymbolStd{\sigma}
\def\OACsymbol[#1]{\textbf{c}_{#1}}
\def\OACsymbolSize{N}
\def\OACsymbolElement[#1]{c_{#1}}

\def\receivedSequence{\textbf{r}}

\def\superposedOACsymbol[#1]{\textbf{s}_{#1}}

\def\sequenceCategory[#1]{\textit{\textbf{l}}_{#1}}
\def\sequenceCategoryEle[#1]{l_{#1}}
\def\functionRepresentative{\phi}
\def\sequenceCategoryEle[#1]{l_{#1}}

\def\channel[#1]{h_{#1}}
\def\precoder[#1]{p_{#1}}
\def\compositeSymbol{\xi}
\def\compositeModel[#1]{\compositeSymbol_{#1}}

\def\noiseVariance{\sigma_{\text{n}}^2}

\def\targetAlignment{A}

 \def\aSet{\mathbb{M}}
\def\anElementOfSetEle[#1]{{m}_{#1}}
\def\anInjectiveMap{\psi}

\def\CER{\texttt{CER}}

\def\minimumDistance[#1]{d_{\rm min}\left\{{#1,\setOfHistograms[K][Q]}\right\}}
\def\indexIteration{i}
\def\updateRate{\eta}

\def\voronoiRegion[#1]{\mathcal{V}_{#1}}
\def\errorEvent[#1]{E_{#1}}

\def\triggerWaveform[#1]{{\mathtt{T}_{#1}}}
\def\timer[#1]{{\mathtt{C}_{#1}}}

\def\numberOfNodesForDetector{{K_{\text{mu}}}}
\def\setOfNodesForDetector{\mathcal{K}_{\text{mu}}}
\def\oneHotVector[#1]{\textbf{e}_{#1}}
\let\norm\undefined % <-- "Undefine" \norm
\DeclarePairedDelimiter\norm{\lVert}{\rVert}
\newcommand\mydots{\hbox to 1em{.\hss.\hss.}}

\newtheorem{definition}{Definition}

\newtheorem{proposition}{Proposition}

\makeatletter
\newif\ifAC@uppercase@first%
\def\Aclp#1{\AC@uppercase@firsttrue\aclp{#1}\AC@uppercase@firstfalse}%
\def\AC@aclp#1{%
	\ifcsname fn@#1@PL\endcsname%
	\ifAC@uppercase@first%
	\expandafter\expandafter\expandafter\MakeUppercase\csname fn@#1@PL\endcsname%
	\else%
	\csname fn@#1@PL\endcsname%
	\fi%
	\else%
	\AC@acl{#1}s%
	\fi%
}%
\def\Acp#1{\AC@uppercase@firsttrue\acp{#1}\AC@uppercase@firstfalse}%
\def\AC@acp#1{%
	\ifcsname fn@#1@PL\endcsname%
	\ifAC@uppercase@first%
	\expandafter\expandafter\expandafter\MakeUppercase\csname fn@#1@PL\endcsname%
	\else%
	\csname fn@#1@PL\endcsname%
	\fi%
	\else%
	\AC@ac{#1}s%
	\fi%
}%
\def\Acfp#1{\AC@uppercase@firsttrue\acfp{#1}\AC@uppercase@firstfalse}%
\def\AC@acfp#1{%
	\ifcsname fn@#1@PL\endcsname%
	\ifAC@uppercase@first%
	\expandafter\expandafter\expandafter\MakeUppercase\csname fn@#1@PL\endcsname%
	\else%
	\csname fn@#1@PL\endcsname%
	\fi%
	\else%
	\AC@acf{#1}s%
	\fi%
}%
\def\Acsp#1{\AC@uppercase@firsttrue\acsp{#1}\AC@uppercase@firstfalse}%
\def\AC@acsp#1{%
	\ifcsname fn@#1@PL\endcsname%
	\ifAC@uppercase@first%
	\expandafter\expandafter\expandafter\MakeUppercase\csname fn@#1@PL\endcsname%
	\else%
	\csname fn@#1@PL\endcsname%
	\fi%
	\else%
	\AC@acs{#1}s%
	\fi%
}%
\edef\AC@uppercase@write{\string\ifAC@uppercase@first\string\expandafter\string\MakeUppercase\string\fi\space}%
\def\AC@acrodef#1[#2]#3{%
	\@bsphack%
	\protected@write\@auxout{}{%
		\string\newacro{#1}[#2]{\AC@uppercase@write #3}%
	}\@esphack%
}%
\def\Acl#1{\AC@uppercase@firsttrue\acl{#1}\AC@uppercase@firstfalse}
\def\Acf#1{\AC@uppercase@firsttrue\acf{#1}\AC@uppercase@firstfalse}
\def\Ac#1{\AC@uppercase@firsttrue\ac{#1}\AC@uppercase@firstfalse}
\def\Acs#1{\AC@uppercase@firsttrue\acs{#1}\AC@uppercase@firstfalse}

\acrodef{WSN}{wireless sensor network}
\acrodef{USRP}{universal software radio peripheral}
\acrodef{SN}{sensor node}
\acrodef{FC}{fusion center}
\acrodef{MAC}{multiple-access channel}
\acrodef{FL}{federated learning}
\acrodef{ED}{edge device}
\acrodef{CS}{compressed sensing}
\acrodef{ES}{edge server}
\acrodef{DCN}{data center network}
\acrodef{RIS}{reconfigurable intelligent surfaces}
\acrodef{IMC}{in-memory computing}
\acrodef{FPGA}{field-programmable gate array}
\acrodef{SDR}{software-defined radio}
\acrodef{PS}{processing system}
\acrodef{SS}{soft synchronization}
\acrodef{IQ}{in-phase/quadrature}
\acrodef{IP}{intellectual property}
\acrodef{DMA}{direct-memory access}
\acrodef{RAM}{random access memory}
\acrodef{CC}{companion computer}
\acrodef{FEE}{function estimation error}
\acrodef{MSK}{minimum-shift keying}
\acrodef{TDMA}{time-domain multiple access}
\acrodef{PLNC}{physical-layer network coding}
\acrodef{UAV}{unmanned aerial vehicle}
\acrodef{LoRa}{Long-Range}
\acrodef{DC}{direct-current}
\acrodef{DAC}{digital-to-analog converter}
\acrodef{ADC}{analog-to-digital converter}
\acrodef{CS}{complementary sequence}
\acrodef{GCP}{Golay complementary pair}
\acrodef{ANF}{algebraic normal form}
\acrodef{AACF}{aperiodic auto-correlation function}
\acrodef{RM}{Reed-Muller}
\acrodef{MOCZ}{modulation on conjugate-reciprocal zeros}
\acrodef{BMOCZ}{binary modulation on conjugate-reciprocal zeros}
\acrodef{dizet}[DiZeT]{direct zero-testing}
\acrodef{PPDU}[PPDU]{physical-layer protocol data unit}

\acrodef{PCP}[PCP]{phase-coded pilot}
\acrodef{GPS}{Global Positioning System}
\acrodef{WSN}{wireless sensor network}
\acrodef{CORDIC}{coordinate rotation digital computer}
\acrodef{LAN}{local area network}
\acrodef{TDM}{time-domain multiplexing}

\acrodef{PUCCH}{physical uplink control channel}
\acrodef{PRACH}{physical random access channel}

\acrodef{OBO}{output-power back-off}
\acrodef{ACLR}{adjacent-channel-leakage ratio}

\acrodef{LDPC}{low-density parity check}

\acrodef{PDF}{probability density function}
\acrodef{CDF}{cumulative distribution function}
\acrodef{ISI}{inter-symbol interference}

\acrodef{TBMA}{type-based multiple access}

\acrodef{MSFE}{mean-squared function error}
\acrodef{FEE}{function-estimation error}
\acrodef{CER}{computation error rate}
\acrodef{BCER}{block-computation error rate}
\acrodef{CFO}{carrier frequency offset}
\acrodef{TO}{time offset}
\acrodef{PO}{phase offset}
\acrodef{RSSI}{received signal strength  information}

\acrodef{STLC}{space-time line code}
\acrodef{CCI}{co-channel interference}
\acrodef{CSIT}[CSIT]{\ac{CSI} at the transmitter}
\acrodef{CSIR}[CSIR]{\ac{CSI} at the receiver}
\acrodef{MIMO}{multiple-input-multiple-output}
\acrodef{PC}{phase correction}
\acrodef{ZF}{zero-forcing}
\acrodef{ANOVA}{analysis of variance}

\acrodef{PCA}{principal component analysis}
\acrodef{TIG}{Technical Interest Group}

\acrodef{FSK}{frequency-shift keying}
\acrodef{PPM}{pulse-position modulation}
\acrodef{PAM}{pulse-amplitude modulation}

\acrodef{MRC}{maximum-ratio combining}
\acrodef{HP}{hard-coded participation}
\acrodef{HPA}{hard-coded participation with absentees}
\acrodef{SP}{soft-coded participation}
\acrodef{FSK-MV}{\ac{FSK}-based \ac{MV}}
\acrodef{RF}{radio-frequency}
\acrodef{MF}{matched filter}
\acrodef{PPM}{pulse-position modulation}
\acrodef{CSK}{chirp-shift keying}
\acrodef{PPM-MV}[PPM-MV]{\ac{PPM}-based \ac{MV}}
\acrodef{DFT-s-OFDM}{discrete Fourier transform-spread orthogonal frequency-division multiplexing}
\acrodef{SC}{single-carrier}
\acrodef{SGD}{stochastic gradient descent}
\acrodef{signSGD}{sign stochastic gradient descent}

\acrodef{SL}{split learning}
\acrodef{SNR}{signal-to-noise ratio}
\acrodef{RMSE}{root-mean-square error}
\acrodef{OFDM}{orthogonal frequency-division multiplexing}
\acrodef{DFT}{discrete Fourier transform}
\acrodef{PSK}{phase-shift keying}
\acrodef{QAM}{quadrature amplitude modulation}
\acrodef{QPSK}{quadrature phase-shift keying}
\acrodef{PMEPR}{peak-to-mean envelope power ratio}
\acrodef{BER}{bit-error ratio}
\acrodef{SNR}{signal-to-noise ratio}
\acrodef{PSD}{power spectral density}
\acrodef{SE}{spectral efficiency}
\acrodef{CP}{cyclic prefix}
\acrodef{AWGN}{additive white Gaussian noise}
\acrodef{CFR}{channel frequency response}
\acrodef{CIR}{channel impulse response}
\acrodef{MMSE}{minimum mean-squared error}
\acrodef{LMMSE}{linear minimum mean-squared error}
\acrodef{BPSK}{binary phase shift keying}
\acrodef{QPSK}{quadrature phase shift keying}
\acrodef{BLER}{block-error rate}
\acrodef{ML}{machine learning}
\acrodef{MaxLike}[ML]{maximum likelihood}
\acrodef{PHY}{physical layer}
\acrodef{PA}{power amplifier}
\acrodef{UE}{user equipment}
\acrodef{BS}{base station}
\acrodef{IDFT}{inverse discrete Fourier transform}
\acrodef{DoF}{degrees-of-freedom}
\acrodef{IoT}{Internet-of-Things}
\acrodef{FDE}{frequency-domain equalization}
\acrodef{RF}{radio-frequency}
\acrodef{IM}{index modulation}
\acrodef{MF}{matched filter}
\acrodef{PPM}{pulse-position modulation}

\acrodef{MSE}{mean-squared error}
\acrodef{MRT}{maximum-ratio transmission}
\acrodef{ERC}{equal-ratio combining}
\acrodef{BAA}{broadband analog aggregation}
\acrodef{OBDA}{one-bit broadband digital aggregation}
\acrodef{FEEL}{federated edge learning}
\acrodef{FL}{federated learning}
\acrodef{UL}{uplink}
\acrodef{DL}{downlink}
\acrodef{OAC}{over-the-air computation}
\acrodef{TCI}{truncated-channel inversion}
\acrodef{MV}{majority vote}
\acrodef{CNN}{convolution neural network}
\acrodef{ReLU}{rectified-linear unit}
\acrodef{CSI}{channel state information}
\acrodef{PAPR}{peak-to-average power ratio}
\acrodef{SC}{single-carrier}
\acrodef{iid}[IID]{independent and identically distributed}
\acrodef{RMS}{root-mean-square}
\acrodef{4G}{Fourth Generation}
\acrodef{5G}{Fifth Generation}
\acrodef{NR}{New Radio}
\acrodef{LTE}{Long-Term Evolution}
\acrodef{OFDMA}{orthogonal frequency division multiple access}
\acrodef{ICI}{inter-carrier interference}
\acrodef{HARQ}{hybrid automatic repeat request}
\acrodef{D2D}{Device-to-Device}
\acrodef{NOMA}{non-orthogonal multiple access}
\acrodef{OMA}{orthogonal multiple access}

\acrodef{IMT}{International Mobile Telecommunications}
\acrodef{ITU}{International Telecommunication Union}

\acrodef{PDP}{power-delay profile}
\acrodef{TBMA}{type-based multiple access}
\acrodef{CDSF}{categorically-distinct symmetric function}

\begin{document}

\title{A Generic Multi-dimensional Symbol Construction for Digital Over-the-Air Computation and Practical Aspects
\thanks{Alphan~\c{S}ahin is with the Electrical  Engineering Department,
	University of South Carolina, Columbia, SC, USA. E-mail: asahin@mailbox.sc.edu}
}
\author{Alphan \c{S}ahin,~\IEEEmembership{Member,~IEEE}

}

\maketitle

\begin{abstract}
In this paper, we propose a general-purpose multi-dimensional symbol construction for computing an arbitrary symmetric function with digital \ac{OAC} and discuss the practical aspects of coherent aggregation. For our first contribution, we discuss the categorical representation of a symmetric function. By using this representation and leveraging the sufficiency of the histogram to evaluate a symmetric function, i.e., inspired by \ac{TBMA}, we introduce a general approach to design a single set of OAC symbols to compute \textit{any} digital function. For our second contribution, we use a comprehensive platform based on low-cost nodes that maintain synchronization in time, frequency, phase, and amplitude via a trigger mechanism, enabling coherent OAC experiments without \ac{GPS} or cable-based synchronization. Using measurements from the platform, we characterize the phase and amplitude statistics of the composite channel to derive a realistic impairment model for coherent OAC. Through a comprehensive analysis, we demonstrate the effectiveness of the proposed scheme under impairments captured by the proposed model.
\end{abstract}

\begin{IEEEkeywords}
Over-the-air computation, multi-dimensional constellation design, synchronization, software-defined radios
\end{IEEEkeywords}
\acresetall

\section{Introduction}
\Ac{OAC} is a physical layer (PHY) approach that harnesses naturally occurring signal superposition in a  \ac{MAC} to evaluate a mathematical function \cite{Nazer_2007}. With OAC, the fusion node does not acquire information from nodes over orthogonal wireless resources; instead, it aggregates information by overlapping their signals for function computation. Due to its potential benefits in improved resource utilization, it has been considered for a wide range of applications, such as federated learning, distributed localization, and wireless control systems \cite{sahinSurvey2023, perezneira2024waveformscomputingair}. However, despite ongoing research on OAC, there is  no consensus on whether OAC is a viable approach in practice, primarily because of its sensitivity to imperfections and the limited receiver-based processing, as the signal superposition occurs after the distortions.

In the literature, \ac{OAC} schemes can be mainly grouped by the discreteness of the transmitted information and the phase coherency of the aggregation. While the former distinguishes  between digital and analog function computation, the latter concerns how the \ac{CSI} is utilized at the nodes. Digital OAC essentially aims to evaluate a target function over a discrete alphabet that may represent quantized parameters. Hence, it is more compatible with the concepts in digital communications, and can be made resilient against noise in the channel, e.g., through heavy quantization \cite{Sahin_2022MVjournal}, error-correction codes \cite{goldenbaum2015nomographic}, and constellation optimization \cite{Saeed_2024channelComp}. The coherent \ac{OAC}, on the other hand, fundamentally relies on counteracting the distortion in the channel for phase-coherent superposition, whereas non-coherent solutions (see \cite{dahl2026unifiedframeworkunbiasednoncoherent,Goldenbaum_2013tcom,Sahin_2022MVjournal}, and the references therein) use a form of amplitude or index modulation without concerning phase coherency. In this work, we particularly focus on digital and coherent OAC.

One practical concern for digital OAC is that the network may be interested in computing a wide range of functions, and optimizing the OAC symbols can be challenging for each possible function (e.g., see \cite{Saeed_2024channelComp}). One way to address this issue is to use pre- and post- functions based on the Kolmogorov–Arnold representation theorem \cite{kolmogorov:superposition}, as mentioned in \cite{Saeed_sumComp,Saeed_ICASSPqam,azimiabarghouyi2026outofaircomputationenablingstructured}. However, designing these functions in general is not trivial. Furthermore, their non-linearity can degrade performance in the presence of noise. Secondly, it is desirable to use digital OAC with coherent superposition, as coherent superposition can significantly outperform its non-coherent counterpart in ideal conditions \cite{Yejin_2024}. However, it can be argued that its sensitivity to imperfections may offset these performance gains in practice. We then ask two fundamental questions regarding the digital coherent OAC: 
\begin{itemize}
	\item Is it possible to construct a \textit{single} generic set of \ac{OAC} symbols that is suitable for computing \textit{all} possible functions with a low \ac{CER}?
	\item What is the impairment model for coherent OAC if one considers the underlying communication protocol, mobility, hardware imperfections, and estimation errors?
\end{itemize}

To address the first question, we propose a multi-dimensional OAC symbol construction. We first discuss the categorical representation of a symmetric function to highlight the equivalency of different target functions. 
By leveraging the sufficiency histogram of symbols for evaluating a symmetric function, we target a construction in which the superposed OAC symbols enable the decoder to identify the histogram of the symbols and compute \textit{all} categorical representatives. By generalizing our earlier construction in \cite{sahinMILCOM2026}, inspired by \ac{TBMA} \cite{Mergen_2006tsp}, via a map that gives a unique representation over a set, we obtain an OAC symbol construction that can be adapted to various dimensions to achieve a higher computation rate. We then assess the trade-off between minimum distance, number of nodes, resource utilization, and parameter size, and derive the corresponding union bound for the \ac{CER}.

To address the second question, we develop a comprehensive testbed using low-cost off-the-shelf \acp{SDR} (i.e., Adalm Pluto) and host computers, where the nodes maintain time, frequency, phase, and amplitude synchronization through custom communication protocols while retaining  the flexibility of \acp{SDR}. For time and phase synchronization, we implement a flexible trigger  method in the \ac{FPGA} of the \acp{SDR} and use the \ac{PCP} strategy \cite{sahin_PIMRC2025}. By using this testbed, we obtain an impairment model for amplitude and phase distortions based on realistic measurements, without using any cable or \ac{GPS}-based synchronization. Furthermore, we develop an enhanced decoder that relies on a joint multi-user and histogram detection to mitigate \ac{CER} under impairments.

\subsection{Related Work}
Digital OAC has been studied extensively in various forms under different assumptions in the literature. For example, traditional constellations such as \ac{BPSK}, \ac{PAM}, and \ac{QAM}, along with several encoding techniques, are employed to compute \ac{MV} or sum functions in \cite{Guangxu_2021, Minjie_2022, Jha_2021, Liang_2022}. Mapping strategies for a generalized QAM constellation for the sum function are theoretically analyzed in \cite{Saeed_sumComp}. In \cite{Saeed_2024channelComp}, the OAC constellation design on the complex plane with the goal of computing functions beyond summation is studied, and several optimization techniques are explored for a given target function. The introduced techniques in \cite{Saeed_2024channelComp} are generalized to multiple dimensions in \cite{Xiaojing_2026}.

One way of extending 2D constellations to multiple dimensions is to exploit the representation of a parameter in a different number system. For instance, in \cite{Xiugang_2016}, a binary representation of parameters is considered, and each bit in the representation is mapped to a \ac{BPSK} symbol to compute the sum. In \cite{sahin2022md}, a balanced number system with \ac{FSK} is utilized for signed parameters, and the framework is extended to binary representation by using two's complement in \cite{Wang_2026}. Another important framework is \ac{TBMA} \cite{Mergen_2006tsp}. Originally, \ac{TBMA} was proposed for parameter estimation based on the histogram of transmitted OAC symbols. It can be implemented via multi-dimensional modulation techniques such as \ac{FSK} or \ac{PPM} \cite{Sahin_2022MVjournal}. It has been used in many studies to evaluate statistical measures (e.g., see \cite{Martinez_2025} and the discussion in \cite{perezneira2024waveformscomputingair}).

To increase the reliability of the computation, OAC is often considered with a lattice code. In fact, Nazer's pioneering work is based on a lattice code and uses modulo-lattice modulation, a \textit{semi-analog} scheme \cite{Kochman_2009} repeatedly with increasingly finer updates (see the proof of \cite[Theorem~3]{Nazer_2007} and also  \cite{Lan_2023}). In \cite{goldenbaum2015nomographic} and \cite{Azimi_2024}, a set of parameters is mapped to a point in a lattice for computation. A hierarchical transmission of the components of a nested lattice is proposed for digital function computation in  \cite{azimiabarghouyi2026outofaircomputationenablingstructured}. 
%We also note that the histogram in TBMA can be interpreted as a lattice point, allowing the fusion node to mitigate noise, as discussed in \cite{martinezgost2026exponentialnoiserobustnesstypebased}. 

In the state of the art, there are discussions regarding coherent and non-coherent aggregation, and their trade-offs (e.g., see \cite{sahinSurvey2023, Yejin_2024,dahl2026unifiedframeworkunbiasednoncoherent}, and the references therein). However, ultimately, OAC's performance is closely tied to its implementation and the underlying protocol. Although real-world demonstrations of OAC are available in the literature, e.g., see \cite{Liang_2022,  Kortke_2014, abari_2016oac, Guo_2021,sahinGC_2022, Xie_2022blockchain,zhu2025timelyparameterupdatingovertheair}, the results cannot be easily generalized or repurposed to assess other schemes. Furthermore, many demonstrations rely on cable-based synchronization, which weakens the plausibility of the experiments. The primary difficulty in demonstrating OAC is the need for synchronized \ac{UL} transmissions, which many off-the-shelf low-cost \acp{SDR} do not support. The researchers address this issue via custom designs, e.g., \cite{Abari2015Airshare, Xie_2022blockchain, Guo_2021, sahinGC_2022}, while introducing techniques to address impairments such as \ac{PO} and \ac{CFO}. Nonetheless, there is still a major gap in replicability and in flexible, low-cost testbeds, and, to the best of our knowledge, in the impairment models supported by realistic measurements for assessing an OAC scheme in practice.

\def\alength{N}
\def\vectorX{\textbf{x}}
{\em Notation:} The sets of complex numbers, real numbers, and integers are denoted by $\complexNumbers$,  $\realNumbers$, and $\integers$, respectively. The set of $\{0,1,\mydots,\alength-1\}$ is denoted by $[\alength]$.
%%The function $\signNormal[\cdot]$ results in $1$, $-1$, or $0$ for a positive, a negative, or a zero-valued argument, respectively. 
The $\alength$-dimensional all-zero vector and an $\alength\times \alength$ identity matrix are  $\zeroVector[{\alength}]$ and $\identityMatrix[{\alength}]$, respectively. 
%%The function $\indicatorFunction[\cdot]$ results in $1$ if its argument holds, otherwise, it is $0$. 
$\expectationOperator[\cdot][]$ is the expectation over all random variables. 
%%The operation $\diagOperator[\textbf{a}]$  returns a square diagonal matrix with the elements of vector \textbf{a} on the main diagonal.
%%$ \nabla \lossFunctionSample[{\modelParameters}]$ denotes the gradient of the function $f$, i.e. $\nabla f$, at the point $\modelParameters$. 
The zero-mean circularly symmetric multivariate  complex Gaussian distribution with the covariance matrix ${\textbf{\textrm{C}}_{\alength}}$ of an $\alength$-dimensional random vector $\vectorX\in\complexNumbers^{\alength\times1}$ is denoted by
$\vectorX\sim\complexGaussian[\zeroVector[\alength]][{\textbf{\textrm{C}}_{\alength}}]$.
%%The binomial distribution with the $K$ trials and the success probability  $p$ for each trial is $\binomDist[K][p]$.
%The uniform distribution with the support between $a$ and $b$ is $\uniformDistribution[a][b]$. 
%%Normal distribution with mean $\mu$ and variance $\sigma^2$ is $\gaussianDist[\mu][\sigma^2]$.
The Euler's number and $\sqrt{-1}$ are denoted by $\constante$ and   $\constanti$, respectively.
The complex conjugate of $x=a+\constanti b$ is $x^*=a-\constanti b$.
The arithmetic mean of the samples $\{x_k\}$ is denoted as $\arithmeticMean[x]$.
The $\ell_2$-norm of the vector $\vectorX$ is $\norm{\vectorX}_2$. 
%The cumulative distribution function of the standard normal distribution is $\normalCDF[\cdot]$. 
The function $\indicatorFunction[\cdot]$ results in $1$ if its argument holds; otherwise, it is $0$.
The probability of  an event $A$ is denoted by $\probability[A]$. 
%The conditional probability of an event $A$ given the event $B$ is shown as $\probability[A|B]$. $\oneVector[L]$ and $\zeroVector[L]$ denote the vectors of length $L$, where their elements are only $1$ or $0$, respectively.

\section{System Model}

Consider a scenario with $\numberOfEdgeDevices$ nodes and one fusion node with the goal of evaluating a target function $ \function:\functionDomain\rightarrow\functionRange\subset\realNumbers$, where $\setOfParameters\triangleq\{\parameter[0],\mydots\parameter[\numberOfParameters-1]|\parameter[\indexParameters]\in\realNumbers,\forall\indexParameters\in[\numberOfParameters]\}$ is the parameter space common for all nodes and $\functionRange\triangleq\{\fcnOutput[0],\mydots\fcnOutput[\numberOfFcnOutputs-1]|\fcnOutput[\indexFcnOutput]\in\realNumbers,\forall\indexFcnOutput\in[\numberOfFcnOutputs]\}$ is the range of the target function. 
Let $\functionArgument[\indexED]\in\setOfParameters$ and $\functionOutput\in\functionRange$ denote the parameter of the $\indexED$th node and the image of $\functionArgumentSequence\triangleq(\functionArgument[1],\mydots,\functionArgument[\numberOfEdgeDevices])$ under $\function$, i.e., $\functionOutput=\functionArbitrary[][{\functionArgument[1],\dots,\functionArgument[\numberOfEdgeDevices]}]$, respectively. In this study, we assume that  the target function is symmetric, i.e., $\functionArbitrary[][{\functionArgument[1],\dots,\functionArgument[\numberOfEdgeDevices]}]=\functionArbitrary[][{\functionArgument[\pi(1)],\dots,\functionArgument[\pi(\numberOfEdgeDevices)]}]$ for every permutation $\pi$ of $[1,\mydots,\numberOfEdgeDevices]$. Several examples of symmetric functions are maximum, minimum, majority vote, threshold function with uniform weights, sum, arithmetic mean, median, and product.

To obtain a framework independent from the target function's domain and range,
let $\mappingParameterToSymbol:\setOfParameters\rightarrow\setOfSymbols$ and  $\mappingOutputToCategory:\functionRange\rightarrow\setOfCategories$ be bijective mappings of the parameter space $\setOfParameters$ to the symbol space $\setOfSymbols\triangleq\{0,1,\mydots,\numberOfParameters-1\}$ and
  the range space $\functionRange$ to the category space $\setOfCategories\triangleq\{0,1,\mydots,\numberOfFcnOutputs-1\}$, respectively, with the definitions given by
$\indexParameters\triangleq\mappingParameterToSymbol(\parameter[\indexParameters])$ and	 	$\indexFcnOutput\triangleq\mappingOutputToCategory(\fcnOutput[\indexFcnOutput])$, 
 where $\indexParameters$ and $\indexFcnOutput$ denote a symbol and a category, i.e., the integer representations of the elements of $\setOfParameters$ and $\functionRange$, respectively.
   Let $\symbolSequence\triangleq(\symbol[1],\mydots,\symbol[\numberOfEdgeDevices])$ be a sequence, where  $\symbol[\indexED]=\mappingParameterToSymbol(\functionArgument[\indexED])$ is the symbol at the $\indexED$th node, $\forall\indexED$.
   We then define the \textit{transformed} target function $\functionSymbol:\setOfSymbols^\numberOfEdgeDevices\rightarrow\setOfCategories$ such that 
   \begin{align}
\category=\functionSymbolArbitrary[][{\symbol[1],\mydots,\symbol[\numberOfEdgeDevices]}]=\mappingOutputToCategory(\functionArbitrary[][{\functionArgument[1],\mydots,\functionArgument[\numberOfEdgeDevices]}])=\mappingOutputToCategory(\functionOutput)~.
   \end{align}
We make our discussions based on the transformed target function $\functionSymbol$, unless otherwise stated.

We assume that all nodes apply a common encoding procedure. To model this procedure, let $\encoder:\setOfSymbols\rightarrow\setOfOACsymbols$ be  a mapping of the symbol space $\setOfSymbols$ to the space of OAC symbols $\setOfOACsymbols\triangleq\{\OACsymbol[0],\mydots,\OACsymbol[\numberOfParameters-1]|\OACsymbol[\indexParameters]=[\OACsymbolElement[\indexParameters,0],\mydots,\OACsymbolElement[\indexParameters,\OACsymbolSize-1]]^{\rm T},\OACsymbolElement[\indexParameters,\indexOACsymbolsElement]\in\complexNumbers,\forall\indexParameters\in[\numberOfParameters],\forall\indexOACsymbolsElement\in[\OACsymbolSize]\}$, respectively, with the definition given by	$\OACsymbol[\indexParameters]\triangleq\encoder(\indexParameters)$,
 where $\OACsymbol[\indexParameters]$ is the $\indexParameters$th complex-valued $\OACsymbolSize$-dimensional \ac{OAC} symbol and $1/\numberOfParameters\times\sum_{\forall\indexParameters}\norm{\OACsymbol[\indexParameters]}_2^2=\OACsymbolSize$, assuming that the OAC symbols are equally likely. For encoding, the $\indexED$th node calculates the corresponding OAC symbol  as
\begin{align}
	\OACsymbol[{\symbol[\indexED]}]=\encoder(\symbol[\indexED])~.
\end{align}
All nodes  simultaneously transmit their OAC symbols over \ac{MAC}. The received sequence at the fusion node is then a superposition of all transmitted OAC symbols and  can be expressed as
\begin{align}
	\receivedSequence = \sum_{\indexED=1}^{\numberOfEdgeDevices} \channel[\indexED]\precoder[\indexED]\OACsymbol[{\symbol[\indexED]}]+\noiseVector= \targetAlignment\sum_{\indexED=1}^{\numberOfEdgeDevices}\OACsymbol[{\symbol[\indexED]}]+\noiseVector~,
	\label{eq:superposition}
\end{align}
where $\noiseVector\sim\complexGaussian[{\zeroVector[\OACsymbolSize]}][{\noiseVariance\identityMatrix[\OACsymbolSize]}]$ is the \ac{AWGN} vector, $\channel[\indexED]\sim\complexGaussian[0][1]$ is the channel coefficient between the fusion node and the $\indexED$th node, and $\precoder[\indexED]\in\complexNumbers$ is the precoder at the $\indexED$th node to counteract the channel distortion. In this work, we consider coherent \ac{OAC} and define the precoder $\precoder[\indexED]$ as $\channel[\indexED]^*/|\channel[\indexED]|^2\times \targetAlignment$, which aligns the amplitudes and phases of the OAC symbols to a predefined amplitude $\targetAlignment$ and $0$~radians, respectively, at the fusion node, leading to the right-hand side of \eqref{eq:superposition}. The value of $\targetAlignment$ can be determined by the fusion node %(e.g., based on the node with the smallest absolute channel coefficient)
 and fed back to the nodes before the superposition. 

In Section~\ref{sec:closingTheGap}, we extend the superposition model in \eqref{eq:superposition}  based on realistic measurements to capture the impact of imperfections on coherent aggregation. In our experiments, we maintain phase and amplitude synchronizations via \ac{PCP} \cite{sahin_PIMRC2025} and a closed-loop \ac{UL} power control, respectively, as we discussed comprehensively in Section~\ref{sec:closingTheGap}.

For decoding, the fusion node performs an operation to determine the function output as $ 	 \categoryDetected=\decoder(\receivedSequence)
 $, 
 where $\decoder:\complexNumbers^\OACsymbolSize\rightarrow\setOfCategories$ is a mapping of $\complexNumbers^\OACsymbolSize$ to the category space $\setOfCategories$ and  $\categoryDetected$ is the detected category. We discuss the decoder $\decoder$  along with the proposed construction in Section~\ref{sec:proposedConstruction}. 
  The detected category is finally mapped to the range of the target function as $\functionOutputDetector=\mappingParameterToSymbol^{-1}(\categoryDetected)$.   Finally, in this work, we define the fidelity metric  as  \ac{CER}, i.e.,
 $\CER = \probability[\categoryDetected\neq\category]$.   In \figurename~\ref{fig:systemModel}, we illustrate the constituents of our system model and the corresponding variables based on the discussions in this section.

% The main challenge that we aim to address in this work is  designing a \textit{generic} set of \ac{OAC} symbols that is suitable to compute \textit{all} possible functions with OAC with low \ac{CER}. 

\begin{figure*}[t]
	\centering
	\includegraphics[width = \figuresizeBB]{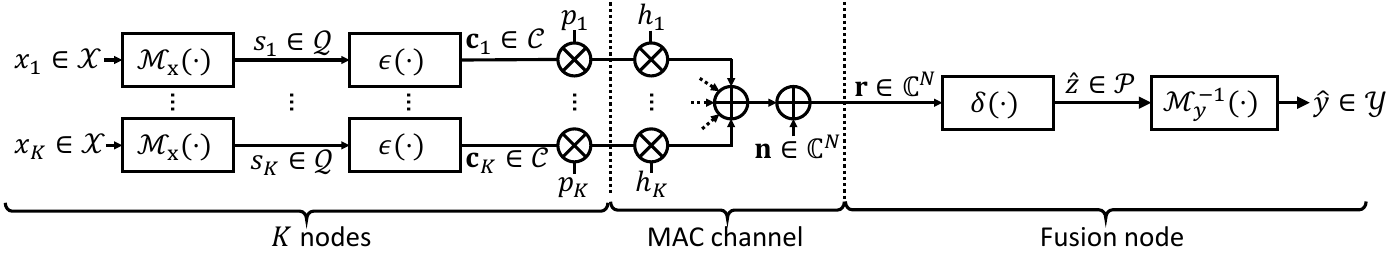}
	\caption{System model.}
	\label{fig:systemModel}
\end{figure*}

\section{Proposed Construction}
\label{sec:proposedConstruction}
In this section, we first discuss the categorical representation of a target function to get further insights into symmetric functions. We then discuss the proposed construction in detail. 

\subsection{Categorical representation of a target function}
The space $\setOfSymbols^\numberOfEdgeDevices$ has a finite number of elements, i.e., $\numberOfParameters^\numberOfEdgeDevices$, implying that the range of $\functionSymbol$, i.e.,  $\setOfCategories$, has also a finite number of elements. However, the symmetric function assumption further restricts $\setOfCategories$. This is because, to calculate the output of a symmetric function, one does not need to know the node index (i.e., the source/originator of the parameter), but only \textit{the occurrence of each symbol, i.e., the histogram of the symbols}. To make this statement more precise, let $\numberOfEdgeDevicesForAGivenSymbol[\indexParameters]$ be the number of nodes with the symbol $\indexParameters$, i.e., $\sum_{\indexED=1}^{\numberOfEdgeDevices}\indicatorFunction[{\symbol[\indexED]=\indexParameters}]$. 
%The histogram of the symbols is then the sequence  $(\numberOfEdgeDevicesForAGivenSymbol[0],\dots,\numberOfEdgeDevicesForAGivenSymbol[\numberOfParameters-1])$ and 
Then, there exists a function $\functionHistogram$ such that
\begin{align}
\functionSymbolArbitrary[][{\symbol[1],\mydots,\symbol[\numberOfEdgeDevices]}]=\functionHistogramArbitrary[{\numberOfEdgeDevicesForAGivenSymbol[0],\dots,\numberOfEdgeDevicesForAGivenSymbol[\numberOfParameters-1]}]~,
\end{align}
for $\forall\symbol[\indexED]\in\setOfSymbols$ and $ \forall\indexED\in[\numberOfEdgeDevices]$.

Let $\setOfHistograms[\numberOfEdgeDevices][\numberOfParameters]\triangleq\{\histogram[\indexHistogram]|\histogram[\indexHistogram]=[\numberOfEdgeDevicesForAGivenSymbol[\indexHistogram,0],\dots,\numberOfEdgeDevicesForAGivenSymbol[\indexHistogram,\numberOfParameters-1]]^{\rm T},\numberOfEdgeDevicesForAGivenSymbol[\indexHistogram,0]+\cdots+\numberOfEdgeDevicesForAGivenSymbol[\indexHistogram,\numberOfParameters-1]=\numberOfEdgeDevices,\forall\numberOfEdgeDevicesForAGivenSymbol[\indexHistogram,\indexParameters]\in[\numberOfEdgeDevices], r\in[\maxMultiplicty]\}$ be the set of the histograms for given $\numberOfEdgeDevices$ and $\numberOfParameters$.
\begin{proposition}[]
	The cardinality of $\setOfHistograms[\numberOfEdgeDevices][\numberOfParameters]$ is $\maxMultiplicty\triangleq\binom{\numberOfEdgeDevices+\numberOfParameters-1}{\numberOfParameters-1}$.
	\label{prop:maxCardinality}
\end{proposition}
\begin{proof}
The number of distinct histograms for $K$ nodes and $Q$ symbols is equivalent to the number of possible ways to put $K$ indistinguishable  balls into $Q$ distinguishable  bins, which is $\maxMultiplicty$ by the stars and bars method. 
\end{proof}
Due to Proposition~\ref{prop:maxCardinality}, a symmetric function can return at most $\maxMultiplicty$ different values, implying that there exist  at most $\maxMultiplicty$ elements in the category space $\setOfCategories$, i.e., $\numberOfFcnOutputs\le\maxMultiplicty$ (see \cite{Saeed_2024channelComp} for a similar discussion). Note that  the equality holds when the function $\functionHistogram$ results in distinct values for all distinct histograms. However, for most of the symmetric functions, distinct histograms can result in the same value, i.e., they are mapped to the same element in $\setOfCategories$ (e.g., an extreme case is when the function gives the same result for all histograms, i.e., $\numberOfFcnOutputs=1$  (e.g., a constant function)). We formalize this statement by defining the following.

\begin{definition}[Categorical representative]
Let $\functionRepresentative$ be a function that maps the $\indexHistogram$th histogram $\histogram[\indexHistogram]$ to the category  $\sequenceCategoryEle[\indexHistogram]$ such that $\forall\sequenceCategoryEle[\indexHistogram+1]\in[0,\mydots,\max(\sequenceCategoryEle[0],\mydots,\sequenceCategoryEle[\indexHistogram])+1]$ and $\sequenceCategoryEle[0]=0$ (i.e., the sequence $(\sequenceCategoryEle[0],\mydots,\sequenceCategoryEle[\maxMultiplicty-1])$ starts with $0$ (i.e., category 0), and the subsequent element of the sequence can increase by 1 (i.e., a new category), or be any number appeared earlier (i.e., one of the earlier categories)). The function $\functionRepresentative$ is the categorical representative of the function $\functionSymbol$.
\end{definition}

The significance of the categorical representatives is that they capture all symmetric functions for given $\numberOfEdgeDevices$ and $\numberOfParameters$, in the sense that if two distinct target functions share the same categorical representative, they both can be evaluated using the same set of \ac{OAC} symbols.

\begin{proposition}
	The number of categorical representatives equals  the $\maxMultiplicty$th Bell number.
	\label{prop:numberOfFunction}
\end{proposition}

\begin{proof}
%The number of the categorical representatives is identical to the cardinality of  $\{(\sequenceCategoryEle[0],\mydots,\sequenceCategoryEle[\maxMultiplicty-1])|\forall\sequenceCategoryEle[\indexHistogram+1]\in[0,\mydots,\max(\sequenceCategoryEle[0],\mydots,\sequenceCategoryEle[\indexHistogram])+1], \sequenceCategoryEle[0]=0\}$. 
The number of  categorical representatives is equal to the total number of ways to partition a set of  $\maxMultiplicty$ histograms into non-empty, unordered subsets, equal to the $\maxMultiplicty$th Bell number by the definition \cite{Bell01081934}.
\end{proof}

For instance, there are $\maxMultiplicty=3$ distinct histograms, i.e., $\setOfHistograms[\numberOfEdgeDevices][\numberOfParameters]=\{\histogram[0]^{\rm T}=[2,0],\histogram[1]^{\rm T}=[1,1],\histogram[2]^{\rm T}=[0,2]\}$ for $\numberOfEdgeDevices=2$ and $\numberOfParameters=2$. Hence, there are $5$, i.e., the third Bell number, distinct categorical representatives with $(\sequenceCategoryEle[0],\sequenceCategoryEle[1],\sequenceCategoryEle[2])\in\{(0,1,2), (0,1,1), (0,1,0), (0,0,1), (0,0,0)\}$. While  $(0,1,2)$ means that each distinct histogram yields a different function output, $(0,0,0)$ indicates that the function returns the same value for all inputs. For example, for $\setOfParameters=\{\parameter[0]=0,\parameter[1]=-2\}$, the categorical representatives of the target functions $\functionArgument[1]+\functionArgument[2]$, $\min(\functionArgument[1],\functionArgument[2])$, $(\functionArgument[1]+\functionArgument[2])\mod4$, $\functionArgument[1]\functionArgument[2]$,  and ${(\functionArgument[1]+\functionArgument[2])\mod2}$ results in $(0,1,2)$, $(0,1,1)$, $(0,1,0)$, $(0,0,1)$, and $(0,0,0)$, respectively, for $(\sequenceCategoryEle[0],\sequenceCategoryEle[1],\sequenceCategoryEle[2])$.

Proposition~\ref{prop:numberOfFunction} shows that the number of categorical representatives  can be very large even for small values of $\numberOfEdgeDevices$ and $\numberOfParameters$. For example, for $\numberOfEdgeDevices=4$ and $\numberOfParameters=4$, $\maxMultiplicty$ is $35$, and there exist approximately $2.8160e+29$ categorical representatives! Given such a large function space, optimizing OAC symbols for each possible function is computationally intensive.

\subsection{Construction}
To construct a generic set of OAC symbols, a key observation  is that if $\setOfOACsymbols$ can be used to evaluate \textit{the categorical representative with the maximum number of distinct categories} (e.g., the one with $(\sequenceCategoryEle[0],\sequenceCategoryEle[1],\sequenceCategoryEle[2])=(0,1,2)$ for $\maxMultiplicty=3$), it can also be used to evaluate any other categorical representative. This is because  such $\setOfOACsymbols$ must lead to distinct superposed OAC symbols  in $\complexNumbers^\OACsymbolSize$ for distinct histograms, and the distinct points in $\complexNumbers^\OACsymbolSize$ can be mapped to the categories based on the target function. Based on this observation, we seek  $\OACsymbolsMatrix\triangleq[\OACsymbol[0],\mydots,\OACsymbol[\numberOfParameters-1]]\in\complexNumbers^{\OACsymbolSize\times\numberOfParameters}$  such that it maximizes the minimum pairwise Euclidean distance across all possible superposed OAC symbols, i.e., $\{\OACsymbolsMatrix\histogram[\indexHistogram],\forall\indexHistogram\in[\maxMultiplicty]\}$:
\begin{align}
	\OACsymbolsMatrixOpt = \arg\max_{\OACsymbolsMatrix} \minimumDistance[\OACsymbolsMatrix],~\text{s.t.}~ \frac{1}{\numberOfParameters}\sum_{\forall\indexParameters}\norm{\OACsymbol[\indexParameters]}_2^2=\OACsymbolSize~,
	\label{eq:originalProblem}
\end{align}
where
\begin{align}
	\minimumDistance[\OACsymbolsMatrix] \triangleq  \min_{\substack{\forall(\histogram[i],\histogram[j])\in\setOfHistograms[\numberOfEdgeDevices][\numberOfParameters]^2\\\histogram[i]\neq\histogram[j]}} \norm{\OACsymbolsMatrix(\histogram[i]-\histogram[j])}_2~.
	\label{eq:minimumDistance}
\end{align}
Then, the corresponding \ac{MaxLike} decoder can be expressed as
\begin{align}
	\decoder(\receivedSequence)=\functionHistogramArbitrary[\arg\min_{\histogram[]\in\setOfHistograms[\numberOfEdgeDevices][\numberOfParameters]}\norm*{\receivedSequence- \targetAlignment\OACsymbolsMatrixOpt\histogram[]}_2]~.
	\label{eq:MLdetector}
\end{align}
Notice that  \eqref{eq:originalProblem} does not provide a constructive solution. Also, it is not trivial to solve \eqref{eq:originalProblem} for large $\numberOfEdgeDevices$ and $\numberOfParameters$ for a given $\OACsymbolSize$. We address these issues by introducing a specific structure to $\OACsymbolsMatrix$, ensuring the identification of the histogram, as follows.

Let $\basis=(\basisElement[0],\mydots,\basisElement[\numberOfDigits-1])\in\complexNumbers^\numberOfDigits$ and  $\aSet\subseteq\integers^\numberOfDigits$. For $(\anElementOfSetEle[0],\mydots,\anElementOfSetEle[\numberOfDigits-1])\in\aSet$, let $\anInjectiveMap:\aSet\rightarrow\complexNumbers$ be a mapping given by
  \begin{align}
\anInjectiveMap(\anElementOfSetEle[0],\mydots,\anElementOfSetEle[\numberOfDigits-1])=\sum_{\indexDigit=0}^{\numberOfDigits-1}\basisElement[\indexDigit]\anElementOfSetEle[\indexDigit]. 
 \end{align}
 We say that $\basis$ gives \textit{unique representations over $\aSet$} if $\anInjectiveMap$ is injective. For example, if $\basis=\{1,\constanti,\sqrt{2}\}$ and  $\aSet=\integers^3$, $\anInjectiveMap(\anElementOfSetEle[0],\anElementOfSetEle[1],\anElementOfSetEle[2])=\anElementOfSetEle[0]+\anElementOfSetEle[1]\constanti+\anElementOfSetEle[2]\sqrt{2}$ is injective. Thus, every number of the form $\anElementOfSetEle[0]+\anElementOfSetEle[1]\constanti+\anElementOfSetEle[2]\sqrt{2}$ with $\anElementOfSetEle[0],\anElementOfSetEle[1],\anElementOfSetEle[2]\in\integers$ can be uniquely represented. Similarly, $\basis=\{1,1/\numberOfEdgeDevices,1/\numberOfEdgeDevices^2,1/\numberOfEdgeDevices^3\}$ gives unique representations over $[\numberOfEdgeDevices]^4$. 
 This example also illustrates that $\numberOfDigits$ is analogous  to the number of digits in a positional number system. While a large $\numberOfDigits$ increases the resolution, it causes points to get closer under an energy constraint, i.e., an increased sensitivity to noise.

 Now, assume that $\basis$ gives unique representations over $\setOfHistogramsSub\triangleq\{{(\numberOfEdgeDevicesForAGivenSymbol[0],\dots,\numberOfEdgeDevicesForAGivenSymbol[\numberOfDigits-1])}|\numberOfEdgeDevicesForAGivenSymbol[0]+\cdots+\numberOfEdgeDevicesForAGivenSymbol[\numberOfParameters-1]=\numberOfEdgeDevices,\forall\numberOfEdgeDevicesForAGivenSymbol[\indexParameters]\in[\numberOfEdgeDevices]\}$. By setting $\numberOfParameters=\OACsymbolSize\numberOfDigits$,  we propose to construct the OAC symbols as a function of $\basis$ as 
\begin{align}
	\OACsymbolsMatrixBasis[\basis] = \sqrt{\OACsymbolSize} \frac{\OACsymbolsMatrixUnnormalized-\OACsymbolMean}{\OACsymbolStd}~,
	\label{eq:OACsymbolsMatrix}
\end{align}
where 
\begin{align}
 \OACsymbolsMatrixUnnormalized &= \identityMatrix[\OACsymbolSize]\otimes\basis=
  \resizebox{2.1in}{!}{$
  \begin{bmatrix}
 	\begin{matrix}
 		\basisElement[0] & \mydots & \basisElement[\numberOfDigits-1] \\ 
 		0 				 & \mydots & 0\\
 		\vdots&\ddots&\vdots\\
 		0 				 & \mydots & 0
 	\end{matrix}
 	&
 		\mydots
 		&
 	\begin{matrix}
		0 				 & \mydots & 0 \\ 
		0 				 & \mydots & 0\\
		\vdots	&\ddots&\vdots\\
		\basisElement[0] & \mydots & \basisElement[\numberOfDigits-1]
	\end{matrix}
 \end{bmatrix}$~,}
\label{eq:OACsymbolsMatrixUnnormalized}
\end{align}
and 
\begin{align}
	\OACsymbolMean &\triangleq \frac{1}{\numberOfParameters}\sum_{\indexDigit=0}^{\numberOfDigits-1}\basisElement[\indexDigit]~,\label{eq:OACsymbolsMean}\\
	\OACsymbolStd^2 &\triangleq (\OACsymbolSize-1) |\OACsymbolMean|^2 + \frac{1}{\numberOfDigits} \sum_{\indexDigit=0}^{\numberOfDigits-1} |\basisElement[\indexDigit]-\OACsymbolMean|^2~.\label{eq:OACsymbolsStd}
\end{align}
In \eqref{eq:OACsymbolsMatrix}, we centralize  $\OACsymbolsMatrixUnnormalized$ and normalize the columns of  $\OACsymbolsMatrixUnnormalized-\OACsymbolMean$ such that $1/\numberOfParameters\times\sum_{\forall\indexParameters}\norm{\OACsymbol[\indexParameters]}_2^2=\OACsymbolSize$.  As explicitly shown in \eqref{eq:OACsymbolsMatrixUnnormalized}, $\OACsymbolsMatrixUnnormalized$ reuses $\basis$ over $\OACsymbolSize$-dimensional space  and leverages that $\basis$ gives unique representations over $\setOfHistogramsSub$, allowing the decoder to construct the complete histogram by identifying $(\numberOfEdgeDevicesForAGivenSymbol[\indexOACsymbolsElement\numberOfDigits],\dots,\numberOfEdgeDevicesForAGivenSymbol[\indexOACsymbolsElement\numberOfDigits+\numberOfDigits-1])$  from the $\indexOACsymbolsElement$th element of the superposed OAC symbol,  $\forall\indexOACsymbolsElement\in[\OACsymbolSize]$.\footnote{$\OACsymbolsMatrixBasis[\basis]$ can be multiplied by an arbitrary complex number on the unit circle in some implementations without affecting its \ac{CER} performance.}

\subsubsection{Case $\numberOfDigits=1$}
For $\numberOfDigits=1$, the columns of $\OACsymbolsMatrixUnnormalized$ are orthogonal to each other, and we set $\basis=(1)$ without loss of generality. In this case, the OAC symbols after the normalization are the elements of an $\OACsymbolSize$-dimensional simplex signaling scheme \cite{proakis2008digital}, where they are equally spaced and form an equiangular set on $(N-1)$-dimensional plane in $\OACsymbolSize$-dimensional  space. 
For this case, $\OACsymbolMean=1/\OACsymbolSize$ and $\OACsymbolStd^2=(\OACsymbolSize-1)/\OACsymbolSize$, and $\minimumDistance[{\OACsymbolsMatrixBasis[\basis]}]$ for \textit{any} $\numberOfEdgeDevices$ can be calculated as
${\sqrt{2}\OACsymbolSize}/{\sqrt{\OACsymbolSize-1}}$.
For example, the OAC symbols are the columns of $\OACsymbolsMatrix$ given by
\begin{align}
\OACsymbolsMatrix =  \begin{bmatrix}
	+1 & -1\\
	-1 & +1
\end{bmatrix}~,
~\OACsymbolsMatrix =  \sqrt{3}\begin{bmatrix}
	+1 & -\frac{1}{3} & -\frac{1}{3} & -\frac{1}{3}\\
	-\frac{1}{3} & 	+1 & -\frac{1}{3} & -\frac{1}{3}\\
	-\frac{1}{3} & -\frac{1}{3} & 	+1 & -\frac{1}{3}\\
	-\frac{1}{3} & -\frac{1}{3} &  -\frac{1}{3} & 	+1\\
\end{bmatrix}~, \nonumber
\end{align}
for $\OACsymbolSize=2$ and $\OACsymbolSize=4$, and supports $\numberOfParameters=2$ and $\numberOfParameters=4$ parameters, respectively. 

The case $\numberOfDigits=1$ is relevant to \ac{TBMA} \cite{Mergen_2006tsp}. In \ac{TBMA}, each sensor observes a single random variable  and transmits according to the \textit{type} of their observations (e.g., one of the $\numberOfParameters$ quantization levels) over orthogonal resources, and the fusion node estimates the parameter. In contrast, the proposed construction for $\numberOfDigits=1$ leverages histograms for a \textit{digital} function computation, and the decoder detects the histogram via simplex signaling to obtain the category.

\subsubsection{Case $\numberOfDigits=2$}
For $\numberOfDigits=2$, we set $\basis=(1,\constanti)$, which extends the case $\numberOfDigits=1$ from reals to the complex numbers while maximizing $\minimumDistance[{\OACsymbolsMatrixBasis[\basis]}]$. For this case, $\OACsymbolMean=(1+\constanti)/\OACsymbolSize$, $\OACsymbolStd^2=(2\OACsymbolSize-1)/(2\OACsymbolSize)$, and 
$
	\minimumDistance[{\OACsymbolsMatrixBasis[\basis]}]={2\OACsymbolSize}/{\sqrt{2\OACsymbolSize-1}}
$.
For example, for $\OACsymbolSize=1$ and $\OACsymbolSize=2$ (i.e., $\numberOfParameters=2$ and $\numberOfParameters=4$, respectively), the OAC symbols can be calculated as
\begin{align}
	\OACsymbolsMatrix &= 
	 \frac{1}{\sqrt{2}}\begin{bmatrix}
		1-\constanti & 	-1+\constanti\\
	\end{bmatrix}~,
	~\nonumber\\
	\OACsymbolsMatrix &= \frac{\sqrt{3}}{\sqrt{2}}\begin{bmatrix}
		+1-\frac{1\constanti}{3} 				& -\frac{1}{3}+1\constanti 			 & -\frac{1}{3}-\frac{1\constanti}{3} 		& -\frac{1}{3}-\frac{1\constanti}{3}\\
		-\frac{1}{3}-\frac{1\constanti}{3} 		& -\frac{1}{3}-\frac{1\constanti}{3} & +1-\frac{1\constanti}{3} 				& -\frac{1}{3}+1\constanti\\
	\end{bmatrix}~,
 \nonumber
\end{align}
respectively.
This case keeps the properties of the case $\numberOfDigits=1$, however, it halves the resources consumed.

\subsubsection{Case $\numberOfDigits\ge3$}
For this case,  one may use a set of linearly independent numbers over the rationals  to obtain $\basis$ by exploiting irrationals,  transcendental numbers, and primes (e.g., $\{1,\sqrt{p}|\text{$p$ is prime}\}$, $\{\pi^n|n\in\integers\}$, and $\{\constante^{2\pi\constanti\sqrt{p}}|\text{$p$ is prime}\}$) or powers of a radix (e.g., $\{1,1/(\numberOfEdgeDevices+1),1/(\numberOfEdgeDevices+1)^2,1/(\numberOfEdgeDevices+1)^3,\mydots\}$) for a given $\numberOfEdgeDevices$ nodes (see also the discussions in \cite[Ch. 12.4.1]{Zamir_Nazer_Kochman_Bistritz_2014} and the separability condition \cite{Motahari_2014}). Although these strategies have their own merits, i.e., constructive and theoretically guarantee that $\minimumDistance[{\OACsymbolsMatrixBasis[\basis]}]$ is larger than zero, they are often sub-optimal. To improve the performance, we consider the following optimization problem:
\begin{align}
	\basisOpt = \arg\max_{\basis} \minimumDistance[{\OACsymbolsMatrixBasis[\basis]}]~,
	\label{eq:originalProblem2}
\end{align}
and attack the optimization problem in \eqref{eq:originalProblem2} by using a gradient-based approach given by
\begin{align}
\basis^{(\indexIteration+1)} = \basis^{(\indexIteration)} +\updateRate \dfrac{\partial}{\partial\basis}\minimumDistance[{\OACsymbolsMatrixBasis[\basis]}]\Bigg|_{\basis=\basis^{(\indexIteration)}}~,
\label{eq:gaOpt}
\end{align}
where $\updateRate$ is the step size. It is worth noting that  \eqref{eq:gaOpt} inherently takes the centralization and energy normalization into account via \eqref{eq:OACsymbolsMean} and \eqref{eq:OACsymbolsStd} for a given $\basis$.

\begin{table}
	\centering
	\caption{Several $\basis$ and $\minimumDistance[{\OACsymbolsMatrixBasis[\basis]}]$ for given  $\numberOfEdgeDevices$, $\OACsymbolSize$, and $\numberOfDigits$ .}
\begin{tabular}{|@{}c@{}|@{}c@{}|@{}c@{}|@{}c@{}|}
	\hline$\numberOfEdgeDevices$ & $(\OACsymbolSize,\numberOfDigits)$ & $ \Re\{\basis^{\rm T}\},\Im\{\basis^{\rm T}\}$ & $d_{\rm min}$ \\ 
	\hline$2$ & $(1,2)$ & {[1,0]}, {[0,1]} & 2 \\
	\hline
	$2$ & $(1,4)$ & {[1,-0.3724,-0.6022,0.0085]}, {[0,-1.0581,0.6595,-0.1328]} & 1.15 \\
	\hline
	$2$ & $(1,8)$ &\begin{tabular}{c} {[1,-0.279,-0.3458,-0.9852,-0.2029,0.5705,-0.6993,0.7903]},\\ {[0,-1.2224,0.1261,-0.4851,0.6158,-1.469,0.4943,0.8584]} \end{tabular}& 0.49 \\
	\hline
	$2$ & $(2,1)$ & {[1]}, {[0]} & 2.83 \\
	\hline
	$2$ & $(2,2)$ & {[1,0]}, {[0,1]} & 2.31 \\
	\hline
	$2$ & $(2,4)$ & {[1,-0.718,0.1916,0.4636]}, {[0,-0.6552,1.1367,-0.8545]} & 1.41 \\
	\hline
	$2$ & $(4,1)$ & {[1]}, {[0]} & 3.27 \\
	\hline
	$2$ & $(4,2)$ & {[1,0]}, {[0,1]} & 3.02 \\
	\hline
	$2$ & $(8,1)$ & {[1]}, {[0]} & 4.28 \\
	\hline
	$4$ & $(1,2)$ & {[1,0]}, {[0,1]} & 2 \\
	\hline
	$4$ & $(1,4)$ & {[1,-0.0365,0.2503,-0.6319]}, {[0,-0.2599,0.793,-0.0824]} & 0.87 \\
	\hline
	$4$ & $(1,8)$ &\begin{tabular}{c} {[1,0.3035,-0.3167,-0.3981,0.2938,-0.0465,-0.4866,-0.0451]},\\ {[0,0.4964,-0.7112,0.166,-0.9521,0.3839,0.0081,0.5907]} \end{tabular}& 0.14 \\
	\hline
	$4$ & $(2,1)$ & {[1]}, {[0]} & 2.83 \\
	\hline
	$4$ & $(2,2)$ & {[1,0]}, {[0,1]} & 2.31 \\
	\hline
	$4$ & $(2,4)$ & {[1,0.7329,-0.5872,-0.2982]}, {[0,-0.7216,0.2572,-1.3622]} & 0.98 \\
	\hline
	$4$ & $(4,1)$ & {[1]}, {[0]} & 3.27 \\
	\hline
	$4$ & $(4,2)$ & {[1,0]}, {[0,1]} & 3.02 \\
	\hline
	$4$ & $(8,1)$ & {[1]}, {[0]} & 4.28 \\
	\hline
\end{tabular}
	\label{table:optBasis}
\end{table}

In our implementation of \eqref{eq:gaOpt}, we use automatic  differentiation to calculate the gradient by using \eqref{eq:minimumDistance} and \eqref{eq:OACsymbolsMatrix}.  We use 2000 iterations and decrease $\updateRate$ from $0.1$ to $0.001$ linearly over the iterations. We run the optimizer with $1000$ randomly initialized $\basis^{(0)}$, sampled from the standard Gaussian distribution. We choose the best local optimum over these solutions.

In \tablename~\ref{table:optBasis}, we provide several $\basis$ and $\minimumDistance[{\OACsymbolsMatrixBasis[\basis]}]$  for  $\numberOfEdgeDevices\in\{2,4\}$ and various $(\OACsymbolSize,\numberOfDigits)$ pairs. The corresponding multi-dimensional constellations are also illustrated in \figurename~\ref{fig:constellationK2} and \figurename~\ref{fig:constellationK4}. Notice that the multi-dimensional OAC constellations are structured for $\numberOfDigits=1$ and $\numberOfDigits=2$. Similarly, the optimization  also yields structured constellations for some cases, e.g., $(\OACsymbolSize,\numberOfDigits)=(1,4)$ for $\numberOfEdgeDevices=2$. A key takeaway from \tablename~\ref{table:optBasis} is the trade-off among reliability, resource consumption, and the number of nodes. While increasing $\numberOfDigits$ reduces the number of resources consumed, it can decrease $\minimumDistance[{\OACsymbolsMatrixBasis[\basis]}]$ considerably for large $\numberOfDigits$ and $\numberOfEdgeDevices$ (e.g., see the superposed symbols in \figurename~\ref{fig:constellationK4}\subref{subfig:K2constellation_n1d8}). On the other hand, for $\numberOfDigits\in\{1,2\}$, $\minimumDistance[{\OACsymbolsMatrixBasis[\basis]}]$ is not a function of $\numberOfEdgeDevices$. This is essentially because these cases do not over-pack the $N$-dimensional space.
Thus,  $\numberOfDigits\in\{1,2\}$ is a more viable choice for a practical scenario  with a large $\numberOfEdgeDevices$, as demonstrated in Section~\ref{sec:numerical}. Nonetheless, the optimal $\basis$ for a large $\numberOfDigits$ still remains a theoretically interesting  question.
\label{subsec:largeD}
\begin{figure}[t]
	\centering
	\subfloat[$(\OACsymbolSize,\numberOfDigits)=(1,2)$.]{\includegraphics[width = \figuresizeSSS]{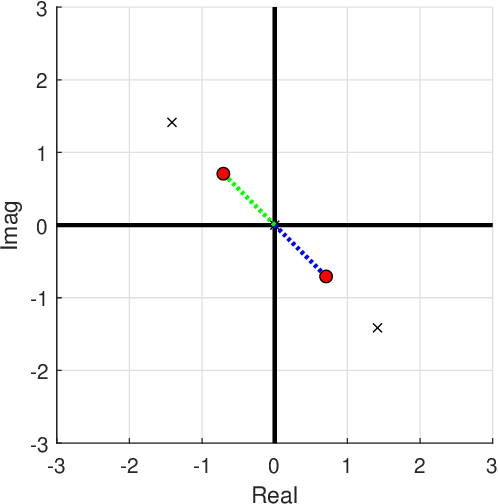}\label{subfig:K2constellation_n1d2}}
	\subfloat[$(\OACsymbolSize,\numberOfDigits)=(1,4)$.]{\includegraphics[width = \figuresizeSSS]{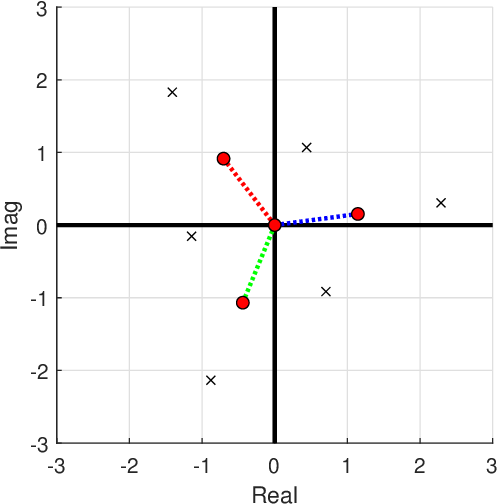}\label{subfig:K2constellation_n1d4}}		
	\subfloat[$(\OACsymbolSize,\numberOfDigits)=(1,8)$.]{\includegraphics[width = \figuresizeSSS]{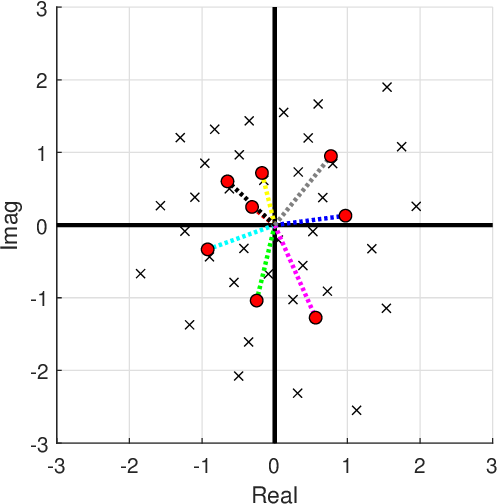}\label{subfig:K2constellation_n1d8}}\\
	\subfloat[$(\OACsymbolSize,\numberOfDigits)=(2,1)$.]{\includegraphics[width = \figuresizeSSS]{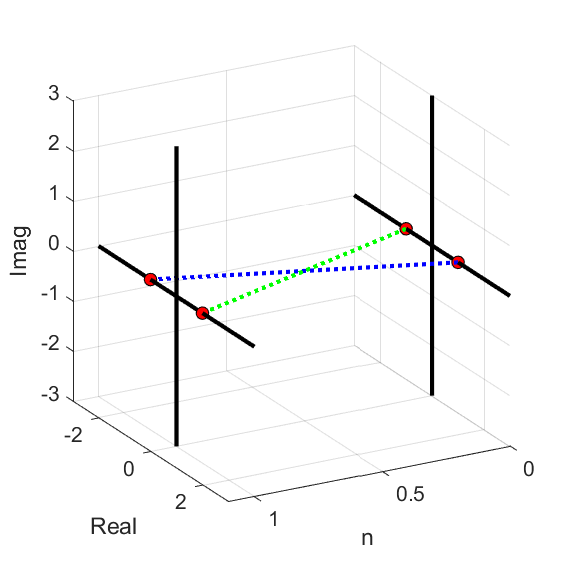}\label{subfig:K2constellation_n2d1}}
	\subfloat[$(\OACsymbolSize,\numberOfDigits)=(2,2)$.]{\includegraphics[width = \figuresizeSSS]{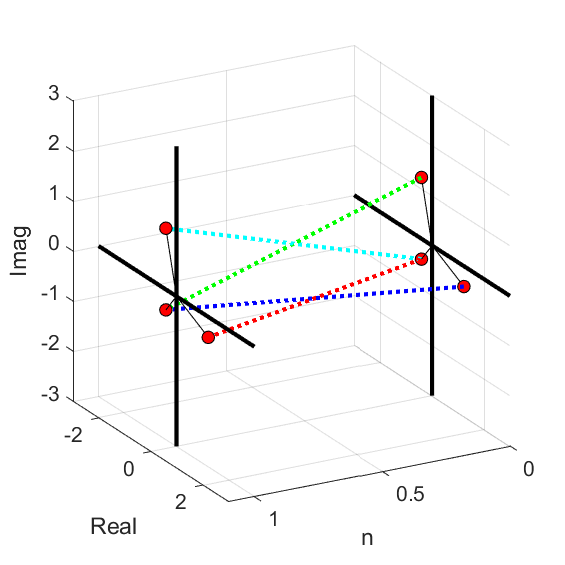}\label{subfig:K2constellation_n2d2}}	
	\subfloat[$(\OACsymbolSize,\numberOfDigits)=(2,4)$.]{\includegraphics[width = \figuresizeSSS]{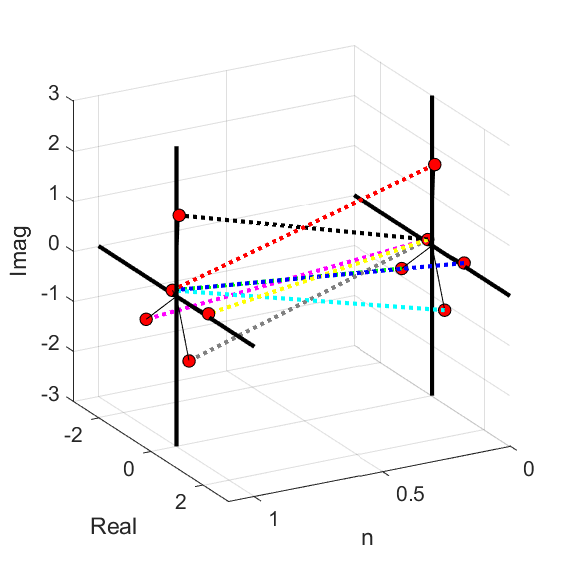}\label{subfig:K2constellation_n2d4}}\\
	\subfloat[$(\OACsymbolSize,\numberOfDigits)=(4,1)$.]{\includegraphics[width = \figuresizeSSS]{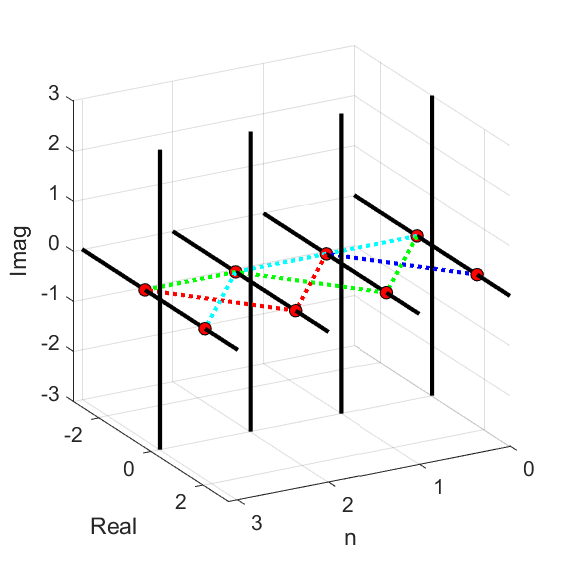}\label{subfig:K2constellation_n4d1}}	
	\subfloat[$(\OACsymbolSize,\numberOfDigits)=(4,2)$.]{\includegraphics[width = \figuresizeSSS]{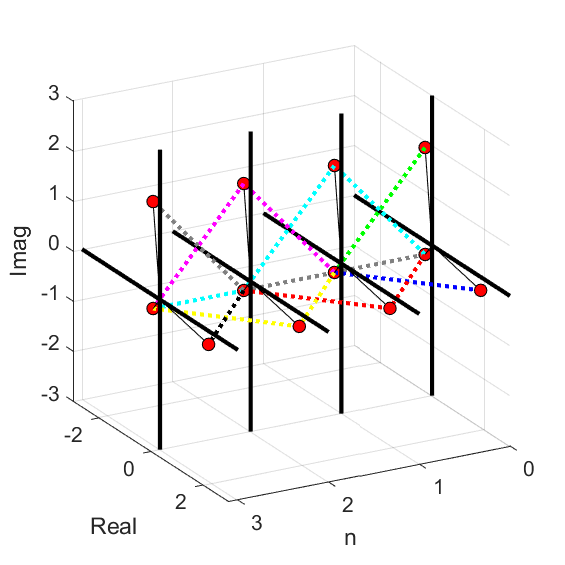}\label{subfig:K2constellation_n4d2}}		
	\subfloat[$(\OACsymbolSize,\numberOfDigits)=(8,1)$.]{\includegraphics[width = \figuresizeSSS]{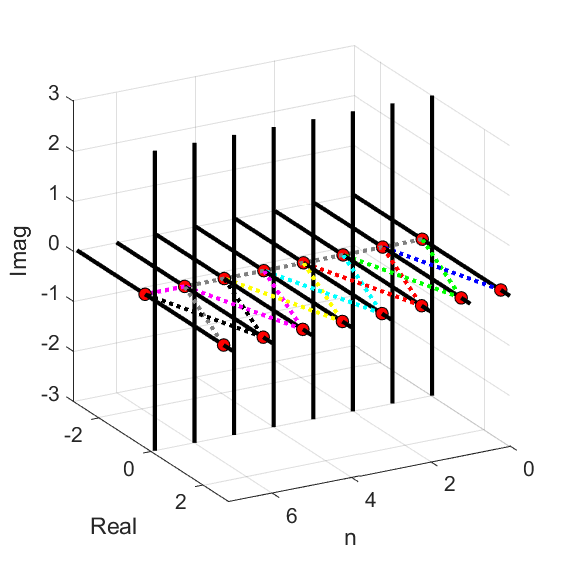}\label{subfig:K2constellation_n8d1}}		
	\caption{Multi-dimensional OAC constellations for computing an arbitrary function  for $\numberOfEdgeDevices=2$. The marker $\times$ shows the superposed symbols.}
	\label{fig:constellationK2}
\end{figure}
\begin{figure}[t]	
	\centering
	\subfloat[$(\OACsymbolSize,\numberOfDigits)=(1,2)$.]{\includegraphics[width = \figuresizeSSS]{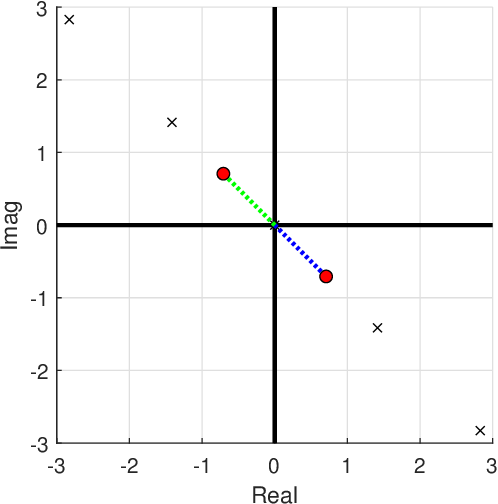}\label{subfig:K4constellation_n1d2}}
	\subfloat[$(\OACsymbolSize,\numberOfDigits)=(1,4)$.]{\includegraphics[width = \figuresizeSSS]{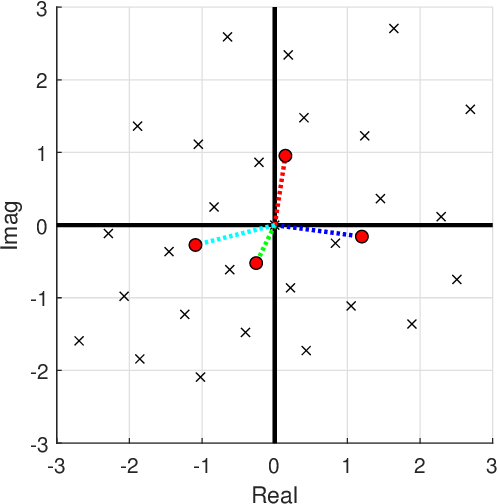}\label{subfig:K4constellation_n1d4}}		
	\subfloat[$(\OACsymbolSize,\numberOfDigits)=(1,8)$.]{\includegraphics[width = \figuresizeSSS]{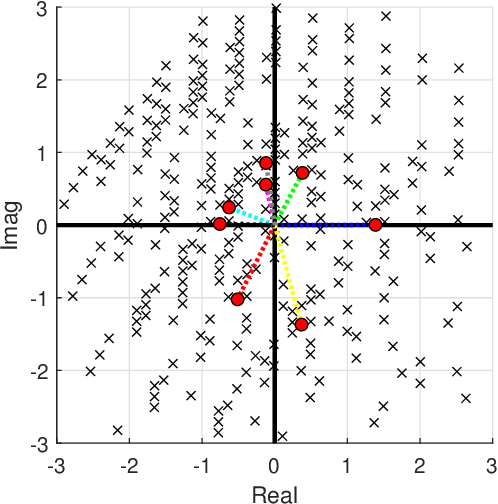}\label{subfig:K4constellation_n1d8}}\\
	\subfloat[$(\OACsymbolSize,\numberOfDigits)=(2,1)$.]{\includegraphics[width = \figuresizeSSS]{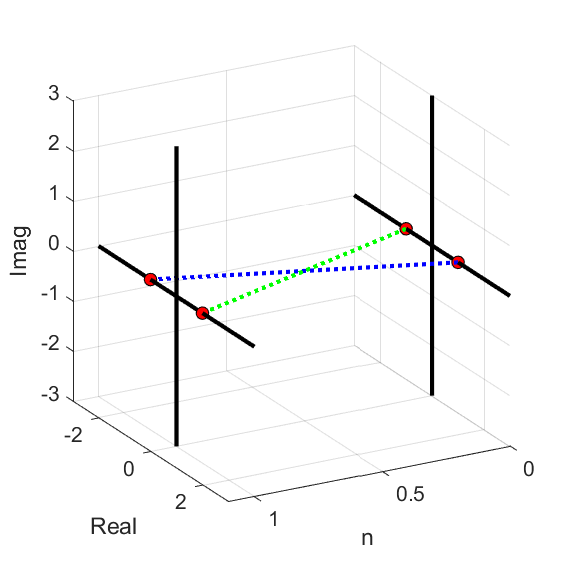}\label{subfig:K4constellation_n2d1}}
	\subfloat[$(\OACsymbolSize,\numberOfDigits)=(2,2)$.]{\includegraphics[width = \figuresizeSSS]{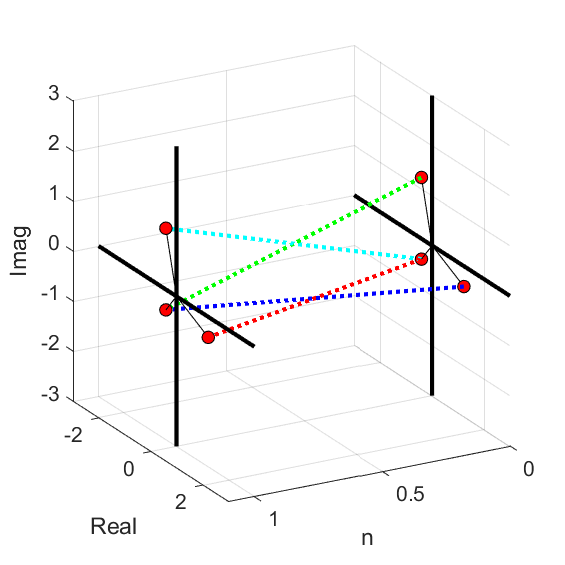}\label{subfig:K4constellation_n2d2}}	
	\subfloat[$(\OACsymbolSize,\numberOfDigits)=(2,4)$.]{\includegraphics[width = \figuresizeSSS]{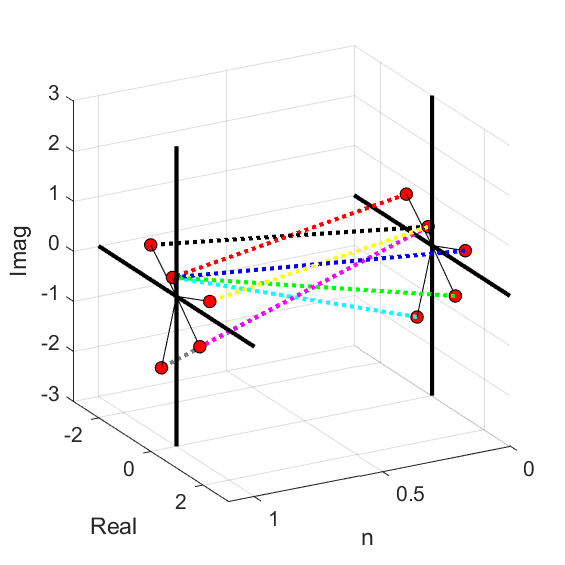}\label{subfig:K4constellation_n2d4}}\\
	\subfloat[$(\OACsymbolSize,\numberOfDigits)=(4,1)$.]{\includegraphics[width = \figuresizeSSS]{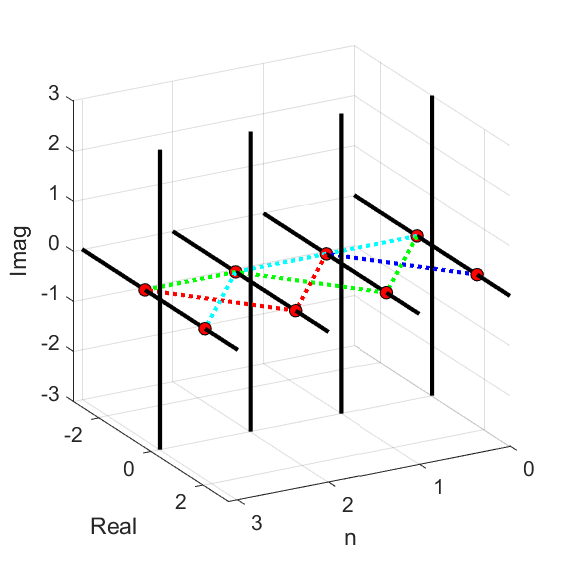}\label{subfig:K4constellation_n4d1}}	
	\subfloat[$(\OACsymbolSize,\numberOfDigits)=(4,2)$.]{\includegraphics[width = \figuresizeSSS]{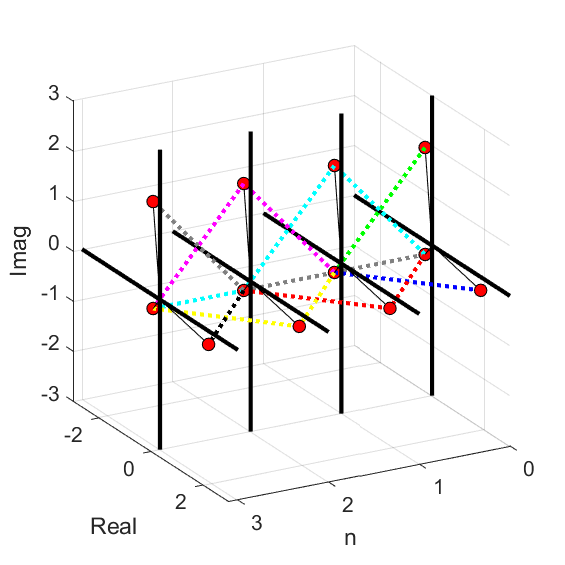}\label{subfig:K4constellation_n4d2}}
	\subfloat[$(\OACsymbolSize,\numberOfDigits)=(8,1)$.]{\includegraphics[width = \figuresizeSSS]{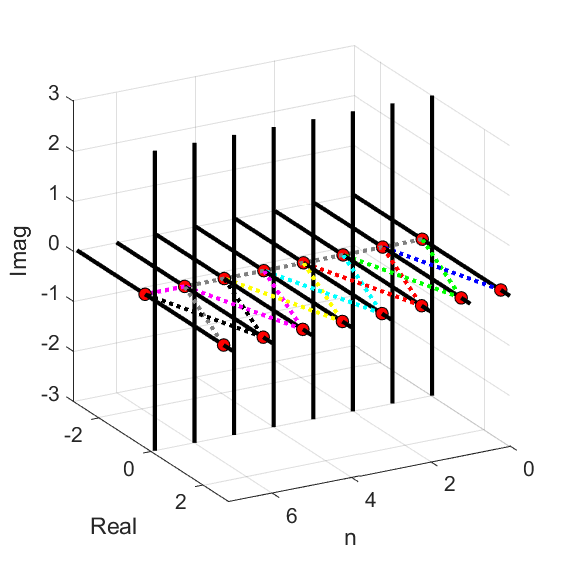}\label{subfig:K4constellation_n8d1}}			
	\caption{Multi-dimensional OAC constellations for computing an arbitrary function  for $\numberOfEdgeDevices=4$. The marker $\times$ shows the superposed symbols.}
	\label{fig:constellationK4}
\end{figure}

\section{Computation Error Rate}
\label{sec:CER}
%Let $\superposedOACsymbol[i]$ be the corresponding superposed OAC symbols at the receiver as $\superposedOACsymbol[i]\triangleq\OACsymbolsMatrix\histogram[i]$. 
Assume that the parameters at the nodes are equally likely and independent. The distribution of $\histogram[]=[\numberOfEdgeDevicesForAGivenSymbol[0],\dots,\numberOfEdgeDevicesForAGivenSymbol[\numberOfParameters-1]]^{\rm T}$ then follows a multinomial distribution with equal parameter probability $1/\numberOfParameters$, i.e.,
\begin{align}
	\probability[{\histogram[]}]=
	%& \frac{\binom{\numberOfEdgeDevices}{\numberOfEdgeDevicesForAGivenSymbol[0]}\binom{\numberOfEdgeDevices-\numberOfEdgeDevicesForAGivenSymbol[0]}{\numberOfEdgeDevicesForAGivenSymbol[1]}\dots\binom{\numberOfEdgeDevices-\sum_{\indexParameters=0}^{\numberOfParameters-2}\numberOfEdgeDevicesForAGivenSymbol[\indexParameters]}{\numberOfEdgeDevicesForAGivenSymbol[\numberOfParameters-1]}}{\numberOfParameters^\numberOfEdgeDevices}=
	\frac{\numberOfEdgeDevices!}{\numberOfEdgeDevicesForAGivenSymbol[0]!\numberOfEdgeDevicesForAGivenSymbol[1]!\dots\numberOfEdgeDevicesForAGivenSymbol[\numberOfParameters-1]!\numberOfParameters^\numberOfEdgeDevices} ~.
	\label{eq:multinomial}
\end{align}
Let  $\voronoiRegion[{\histogram[\indexHistogram]}]$ be the corresponding Voronoi region for the superposition ${\OACsymbolsMatrix\histogram[\indexHistogram]}$. Hence,  the \ac{CER} can be expressed as
\begin{align}
 \CER=\probability[\categoryDetected\neq\category]&=\sum_{\indexHistogram=0}^{\maxMultiplicty-1}\probability[{\text{error}|\histogram[\indexHistogram]}]\probability[{\histogram[\indexHistogram]}]~,
 \label{eq:cerBase}
\end{align}
where
\begin{align}
	\probability[{\text{error}|{\histogram[\indexHistogram]}}]&= 1 - \int_{\voronoiRegion[{\histogram[\indexHistogram]}]} \frac{1}{(\pi\noiseVariance)^{\OACsymbolSize/2}}\exp\left(-\frac{\norm{\textbf{x}-\OACsymbolsMatrix\histogram[i]}^2}{\noiseVariance}\right)d\textbf{x}~, \nonumber
\end{align}
and an error is an event that the detected histogram is different from $\histogram[\indexHistogram]$. We  obtain an upper bound on $\probability[{\text{error}|\histogram[\indexHistogram]}]$ as
\begin{align}
	\probability[{\text{error}|\histogram[\indexHistogram]}] &= \probability[{\bigcup_{\substack{i=0\\i\neq \indexHistogram}}^{\maxMultiplicty-1}{\errorEvent[i]\Bigg|\histogram[\indexHistogram]}}] \le
	\sum_{\substack{i=0\\i\neq \indexHistogram}}^{\maxMultiplicty-1}
	\probability[{\errorEvent[i]|\histogram[\indexHistogram]}] \nonumber\\&= \sum_{\substack{i=0\\i\neq \indexHistogram}}^{\maxMultiplicty-1}\qfunction[\frac{\norm{\OACsymbolsMatrix(\histogram[i]-\histogram[\indexHistogram])}_2}{\sqrt{2\noiseVariance}}]~, \label{eq:condBound}
\end{align}
where $\errorEvent[i]$ is the event that the detected histogram  is $\histogram[i]$. Finally, by using \eqref{eq:multinomial} and  \eqref{eq:condBound} in \eqref{eq:cerBase}, the union bound for the \ac{CER}  can be obtained as
\begin{align}
	 \CER\le \sum_{\indexHistogram=0}^{\maxMultiplicty-1}\sum_{\substack{i=0\\i\neq \indexHistogram}}^{\maxMultiplicty-1}\qfunction[\frac{\norm{\OACsymbolsMatrix(\histogram[i]-\histogram[\indexHistogram])}_2}{\sqrt{2\noiseVariance}}]	\frac{\numberOfEdgeDevices!}{\numberOfEdgeDevicesForAGivenSymbol[\indexHistogram,0]!\dots\numberOfEdgeDevicesForAGivenSymbol[\indexHistogram,\numberOfParameters-1]!\numberOfParameters^\numberOfEdgeDevices} ~.
	 \label{eq:unionBoundCER}
\end{align}

\section{Closing The Gap Between Theory and Practice}
\label{sec:closingTheGap}
Although the superposition model in \eqref{eq:superposition} is useful for assessing the performance of \ac{OAC} schemes relative to one another, it does not consider imperfections. Thus, the \ac{CER} in practice can significantly differ from that based on \eqref{eq:superposition}.  To capture the imperfections, we re-express  \eqref{eq:superposition} as
\begin{align}
	\receivedSequence = \sum_{\indexED=1}^{\numberOfEdgeDevices} \compositeModel[\indexED]\OACsymbol[{\symbol[\indexED]}]+\noiseVector~,
	\label{eq:superpositionPractice}
\end{align}
where $\compositeModel[\indexED]\in\complexNumbers$ is the composite channel response including the precoder $\precoder[\indexED]$, channel $\channel[\indexED]$, and imperfections for the $\indexED$th node. Clearly, if $\compositeModel[\indexED]$ is a constant, e.g., $\compositeModel[\indexED]=\targetAlignment$, $\forall\indexED$, \eqref{eq:superpositionPractice} reduces to \eqref{eq:superposition}. However, due to the combined effect of the imperfections, $\compositeModel[\indexED]$ scatters in the complex plane and deviates from the target amplitude $\targetAlignment$ and the phase of zero radians. 

Let $\PDFimpairments[\compositeSymbol]$ denote the \ac{PDF} of $\compositeModel[\indexED]$, $\forall\indexED$. In practice, $\PDFimpairments[\compositeSymbol]$ is a function of the communication protocol, mobility, hardware imperfections, channel estimation errors, and synchronization errors. Hence, it is not trivial to characterize it analytically. In this section, we address this challenge by introducing an  impairment model by acquiring $\compositeModel[\indexED]$ through realistic measurements \textit{without relying on any auxiliary synchronization (e.g., \ac{GPS} or cable-based synchronization)}. We also discuss a decoding strategy to mitigate the impact of imperfections on coherent superposition.

\begin{figure}[t]
	\centering
	\includegraphics[width =3.5in]{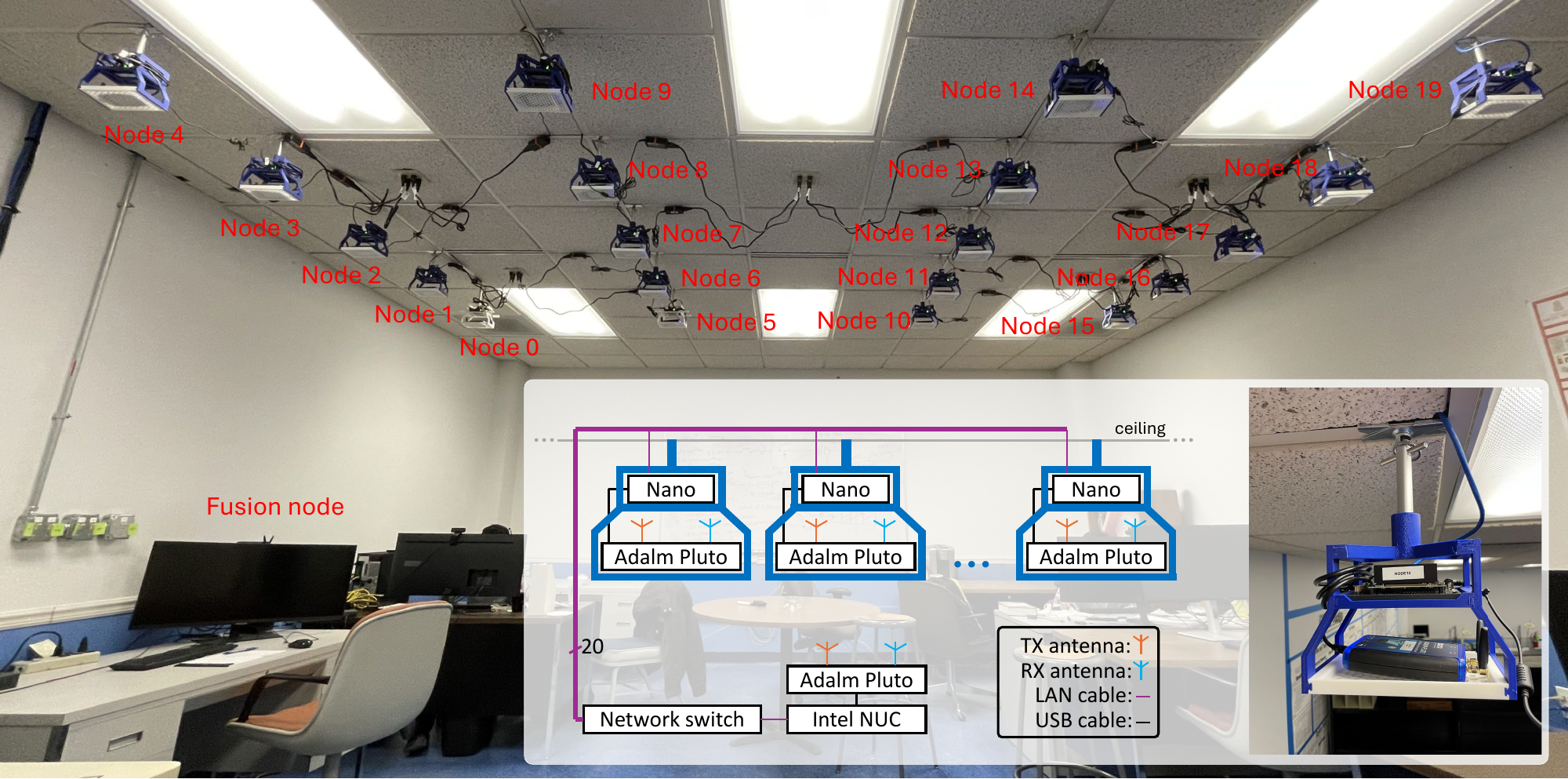}
	\caption{Synchronization testbed at the USC and its structure.}
	\label{fig:testbed}
\end{figure}
\begin{figure*}[t]
	\centering	
	\includegraphics[width =6.0in]{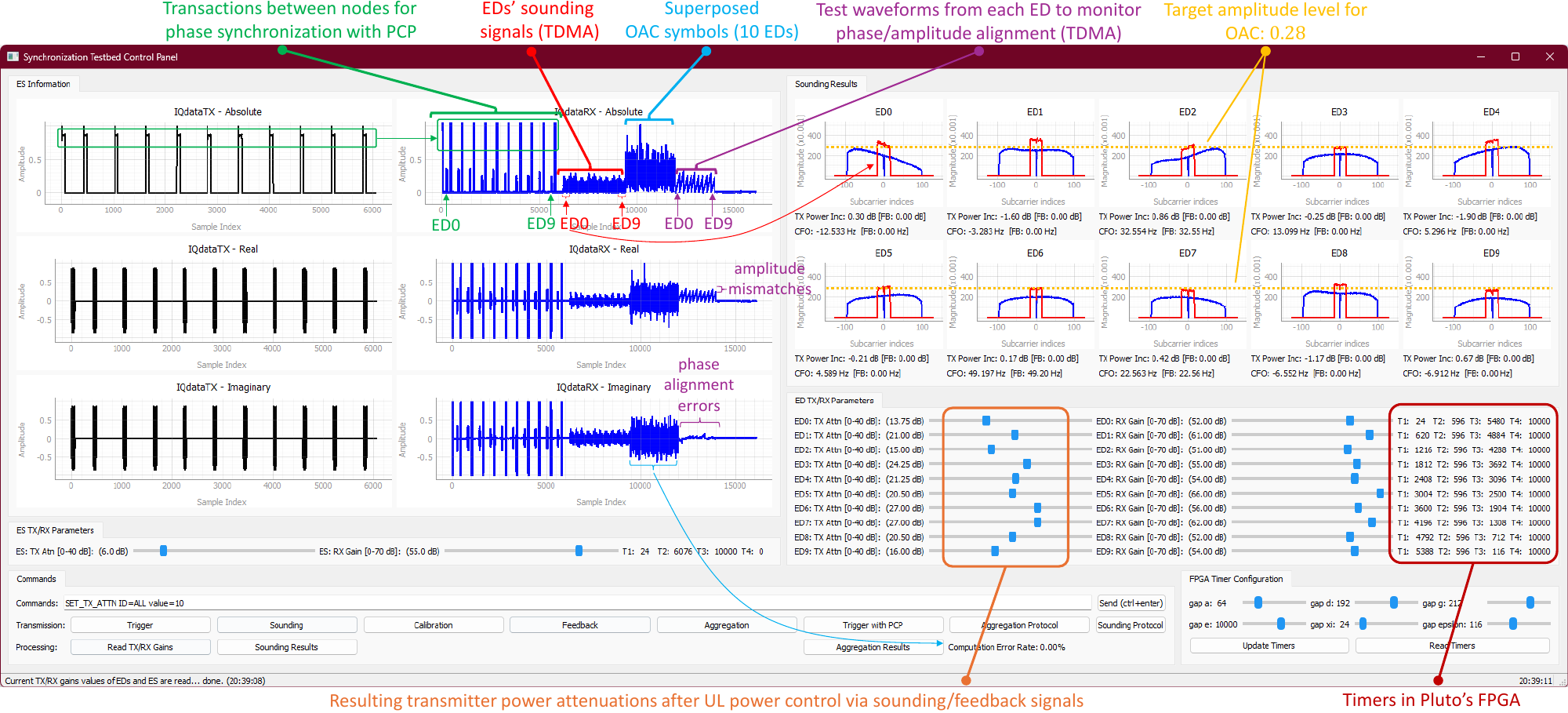}
	\caption{User interface for measuring the impact of impairments on coherent OAC in practice for 10 nodes.}
	\label{fig:userInterface}	
\end{figure*}
\begin{figure*}[t]
	\includegraphics[width =7in]{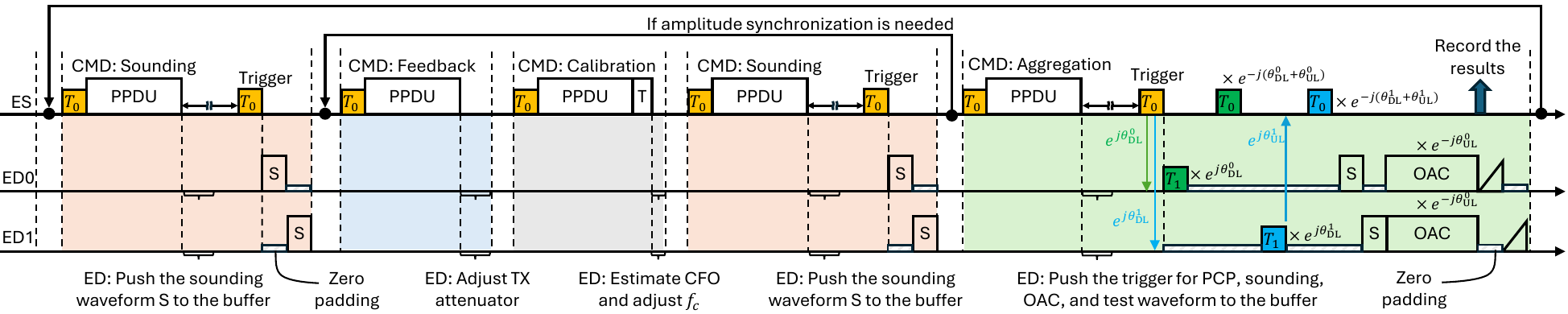}
	\caption{The aggregation protocol for two nodes to extract impairment distribution on coherent aggregation (S: Sounding waveform, T: Trailer waveform,  $\triggerWaveform[1]$, $\triggerWaveform[2]$: Trigger waveforms, $\theta_{\rm DL}^{0,1}$, $\theta_{\rm UL}^{0,1}$: The DL and UL phase mismatches for ED0 and ED1).}
	\label{fig:protocol}	
\end{figure*}
\subsection{Platform and modifications for synchronization}
For our measurements,  we develop a comprehensive platform hosting twenty Adalm Pluto (Rev. C) \acp{SDR} mounted on the ceiling, each representing a node (or an \ac{ED}), and one \ac{SDR} as the fusion node (or \ac{ES}), as shown in \figurename~\ref{fig:testbed}. On the platform, each SDR has a dedicated host computer: NVIDIA Jetson Nano computers for the \acp{ED} and an Intel NUC for the \ac{ES}. We also use an unmanaged network switch to enable communication between the host computers and to upload Python code to the EDs from the ES side. We use a specific user interface, shown in \figurename~\ref{fig:userInterface}, running on the ES side. Through the user interface, we dynamically interact with the ES to initiate the protocol for coherent aggregation and its subroutines, observe the results, monitor the SDR parameters, and acquire data regarding $\compositeModel[\indexED]$. The user interface also allows us to read the parameters of the EDs' SDRs over a \ac{LAN}. However, we do not manually set the parameters for measurement plausibility. 

The time and phase synchronization capability along with the ability to transmit and/or receive IQ data flexibly  with low-cost SDRs is the unique aspect of the platform. Time synchronization fundamentally relies on a trigger mechanism and extends the strategies used in our previous experiments \cite{sahinGC_2022} and \cite{sahin_infocom2023mmWave}. In the testbed, the \ac{FPGA} of each \ac{SDR} is modified to continuously cross-correlate a 64-sample trigger waveform with the IQ samples (either on the transmitter or receiver path). Once the configured trigger waveform is detected, the FPGA starts counting down four configurable timers (i.e., $\timer[1]$, $\timer[2]$, $\timer[3]$, and $\timer[4]$) in units of samples, sequentially. During these countdowns, it enables or disables transmitter and receiver paths based on a configuration. For example, in one configuration, upon the detection instant of the trigger waveform,  an ED's SDR may be configured to acquire the IQ samples from the transceiver (i.e., AD9363) to a buffer in the FPGA (i.e., maximum 16384 samples) in the countdown $\timer[1]$, push the IQ data samples from Pluto's \ac{RAM} to the transceiver via \ac{DMA} in the countdown $\timer[2]$, and wait idle in the countdowns $\timer[3]$ and $\timer[4]$. In this work, we define two configurable trigger signals, denoted as $\triggerWaveform[0]$ and $\triggerWaveform[1]$, whose corresponding coefficients can be loaded into the correlator to switch the trigger waveform during the countdowns.

Another essential aspect is that the correlator can be configured to listen to the transmitter path, enabling dual use of the trigger signal. With this feature, a trigger signal from the ES can be utilized to trigger both the EDs in the network \textit{and} the ES's SDR to start IQ data acquisition. Therefore, the ES can obtain the relevant IQ data when the EDs transmit their signals after being triggered. The carefully chosen timer values, along with zero-padded IQ data to adjust the position of the actual IQ data  (see \figurename~\ref{fig:protocol} for an example), and the dual use of a trigger enable us to maintain time synchronization in the network and timely IQ data acquisition, where we leverage them for both \ac{UL} multiple-access and signal superposition. For instance, the EDs' sounding  for measuring $\compositeModel[\indexED]$ and test waveforms (i.e., a triangle function for visual inspection) are received back-to-back based on \ac{TDMA}, and the OAC signals are superposed with this strategy, as  in \figurename~\ref{fig:userInterface}. 

For phase synchronization, we adopt the \ac{PCP} strategy, introduced in \cite{sahin_PIMRC2025}. The \ac{PCP} is a low-complexity analog feedback method that corrects the total \ac{UL} and \ac{DL} phase mismatch and does not rely on channel reciprocity. In our implementation, we treat the trigger waveform as a preamble and use the aforementioned cross-correlator to estimate the phase rotation in the channel. Whenever the configured trigger waveform is detected, the FPGA registers the correlation output as the channel estimate (i.e., narrowband assumption) and calculates its angle. During the countdowns, the FPGA then rotates the IQ data on the transmitter path by the estimated angle (or its negated value, if configured) to provide feedback or correct the total phase mismatch during aggregation. Since the PCP strategy is implemented in the FPGA and independent of the transmitted IQ data, it is resilient to changes in the propagation environment while retaining the flexibility of SDR.

\figurename~\ref{fig:userInterface} and \figurename~\ref{fig:protocol} show the exchanges for the \ac{PCP} method in the testbed and the corresponding procedure, respectively. In \figurename~\ref{fig:userInterface},  the first portions of the received and transmitted IQ data show the trigger waveforms transmitted from the ES and the EDs for PCP at the ES side.\footnote{Since both transmitter and receiver paths of the ES are active during the PCP exchanges, the transmitted waveforms are also visible at the received IQ data, saturating the ES's receiver as the antennas are close to each other.} The correlators at the EDs and ES are configured to detect the trigger signal and estimate the amount of phase rotation. The procedure is as follows: 1)~The first waveform (slightly cropped in \figurename~\ref{fig:userInterface}) from the ES, i.e., $\triggerWaveform[0]$, triggers the EDs. At this stage, all EDs estimate the DL phase mismatch and push their IQ data to their transceivers. 2)~While the IQ data passes through the transceiver, the first ED's FPGA rotates the trigger signal $\triggerWaveform[1]$ by a factor of the DL phase mismatch. The transmitted trigger signal from the first ED is detected at the ES's correlator. 3)~The negated version of the estimated phase mismatch is then applied to the trigger signal transmitted from the ES. 4)~The first ED then re-estimates the phase error and applies the estimated phase to the sounding, OAC waveform, and the test waveform to achieve phase alignment at the ES location. The exchanges between the ES and other EDs proceed similarly, as shown in \figurename~\ref{fig:userInterface} and \figurename~\ref{fig:protocol}. The phase alignment can also be observed on the triangle waveform transmitted from each ED in \figurename~\ref{fig:userInterface}, where the PCP successfully aligns the phases of the \textit{received} triangular waveform at zero radian (i.e., the in-phase branch) while the imperfections are visible on the quadrature part of the received IQ data.

Note that the transmitted trigger signals in the \ac{UL} and \ac{DL} directions must be different from each other to avoid false alarms when an ED transmits a trigger. Hence, after the first trigger, the ES's FPGA loads the corresponding coefficients for $\triggerWaveform[1]$ and monitors the receiver path. The EDs' correlators always monitor $\triggerWaveform[0]$. Also, the triggers  are chirp signals \cite{sahin_dftsofdmChirp}.

\subsection{Signaling}
The EDs and ES communicate with each other using a custom \ac{OFDM}-based \ac{PPDU} at 920~MHz in both \ac{UL} and \ac{DL}. The PPDU has 192 active subcarriers, uses an \ac{IDFT} size of 256 and a \ac{CP} length of 64, and is based on \ac{BPSK} and a 1/2-rate polar code. We refer the reader to \cite{ashwini_milcom2024} for further details on the \ac{PPDU} structure and modulation/coding parameters. We set the sample rate to 5~Msps for all SDRs.

During measurements, the ES serves as the main controller and coordinates the EDs by sending commands via PPDUs. We define four commands: sounding, feedback, calibration, and aggregation. To indicate the command type, we allocate the first 4 bits in the payload. The following 32 bits, where each bit position is mapped to an ED ID, indicate the recipient EDs of the corresponding command. The EDs continuously monitor the medium and wait for the trigger signal (i.e., $\triggerWaveform[0]$) to acquire IQ data via the aforementioned trigger mechanism in the \ac{FPGA}. Thus, the ES always prepends $\triggerWaveform[0]$ to the PPDUs, as shown in \figurename~\ref{fig:protocol}. After the ED's FPGA detects the trigger signal, the EDs acquire the IQ data from the transceiver. The corresponding host computer then processes the IQ data to decode the PPDU. If it can decode the PPDU and it is a recipient, it acts based on the command type. The commands and the corresponding behaviors can be outlined as follows:
\subsubsection{Calibration command}
We use the calibration command for frequency synchronization in the network. The ES transmits a calibration command along with a trailer signal, which consists of $6$ \ac{OFDM} symbols (each carrying an upsampled Zadoff-Chu sequence of length 97 with a factor of $2$). If an ED is a recipient of the calibration command, it estimates the \ac{CFO} by processing the trailer and re-adjusts its carrier frequency in the UL, since we use the same carrier frequency in the UL and DL. To justify this approach, we also measure the residual \ac{CFO} after compensation through an \ac{UL} sounding. As shown in \figurename~\ref{fig:userInterface}, the estimated \acp{CFO}  in the UL after the correction are within the range of 0-50~Hz. In addition to the CFO correction, the EDs also adjust their receive gains to avoid saturation.
It is worth noting that Adalm Pluto uses a 40~MHz oscillator with a 25 PPM stability rating, which causes significant frequency drift and affects phase alignment. Hence, we initiate calibration any time before aggregation and maintain the residual CFO level in the UL for all EDs less than 50 Hz during the measurements. We do not feed back the residual CFO from the ES to the EDs in this work.

\subsubsection{Sounding command} 
We use the sounding command to measure the \ac{CFR} and the residual \ac{CFO} for each ED to assess frequency and amplitude synchronization. In response to a sounding command, a recipient ED transmits a low-\ac{PAPR} OFDM symbol, repeated four times (each carrying a Golay complementary sequence of length 192). To avoid overlapping, the EDs also prepend zero samples to the corresponding IQ data based on their IDs. Once another trigger signal transmitted from the ES is detected, the FPGA pushes the sounding data to the transceiver, and the ES receives all EDs' sounding signals back-to-back.

\subsubsection{Feedback command}
We use the feedback command to maintain amplitude synchronization across the network. By using the \acp{CFR} obtained from the sounding signals, the ES calculates how much each ED deviates from a target level at the center of the band. If the measured amplitude deviation of an ED exceeds the target value by 0.75 dB, the difference, with a resolution of 0.25 dB, is fed back to the EDs. During the measurements, we set the target level value to be $0.2$. For feedback, we allocate 32 bits (single-precision) per scalar for each recipient ED. Based on received feedback, a recipient ED adjusts the transmitter attenuation factor of Adalm Pluto. In \figurename~\ref{fig:userInterface}, we show the sounding results for each ED, and the absolute value of the CFR at the center of the band for each ED is approximately aligned at 0.2 after feedback. The amplitude mismatches are also visible in the test waveforms.

It is worth noting that the host-based solution is not an ideal approach for amplitude synchronization as it hinders timely feedback. We adopt a host-based approach because the FPGA resources on Pluto are significantly limited and are dedicated to time and phase alignment in our work. Secondly, we notice amplitude variations even when there is no apparent mobility on the platform, which may be due to automatic corrections or DC bias fluctuations in the AD9363 chipset. During the measurements, we did not fix such variations.

\subsubsection{Aggregation command}
We use the aggregation command to instruct an ED to prepare the OAC waveform, including a trigger for PCP, a sounding signal, and a test waveform, write it to RAM, and wait for a trigger signal from the ES. After the aggregation command, the ES transmits trigger waveforms based on the aforementioned PCP strategy, and the EDs push the corresponding IQ data from RAM to their transceivers. The ES receives the superposed OAC symbols along with the sounding and test waveforms.

In this work, we use \ac{DFT-s-OFDM} to transmit the OAC symbols. The  reason for this choice is that time synchronization is not precise and can deviate by $\pm1$ sample duration, causing phase rotations in the frequency domain. Since  phase impairments are detrimental to coherent OAC, we circumvent this issue by using a single-carrier variant, i.e., \ac{DFT-s-OFDM} \cite{sahin_flexibleDFTsOFDM}, where the symbols are carried in the time domain via a Dirichlet sinc kernel, at the expense of \ac{ISI} between OAC symbols. We address the \ac{ISI} by using a smaller \ac{DFT} precoder size (i.e., 32), which extends the main lobe of the Dirichlet sinc kernel, allowing the matched filter to absorb the sample deviations. Note that this approach does not reduce the computation rate, as multiple \ac{DFT} precoders can be stacked in the frequency domain.

During the measurements, we apply $3$~dB and $8$~dB power back-off to the normalized sounding and OAC waveforms, respectively. Since the OAC waveform (32 subcarriers) is 6 times narrower than the sounding waveform (192 subcarriers), the OAC waveform is approximately $1.4$ times stronger at the ES side as compared to the sounding waveform. Thus, the OAC waveform is approximately aligned at $\targetAlignment=0.28$, as shown in the sounding results in \figurename~\ref{fig:userInterface}.

\begin{figure*}[t]
	\centering
	\subfloat[Empirical distribution of $\xi_k$.]{
	\includegraphics[width = \figuresizeSS]{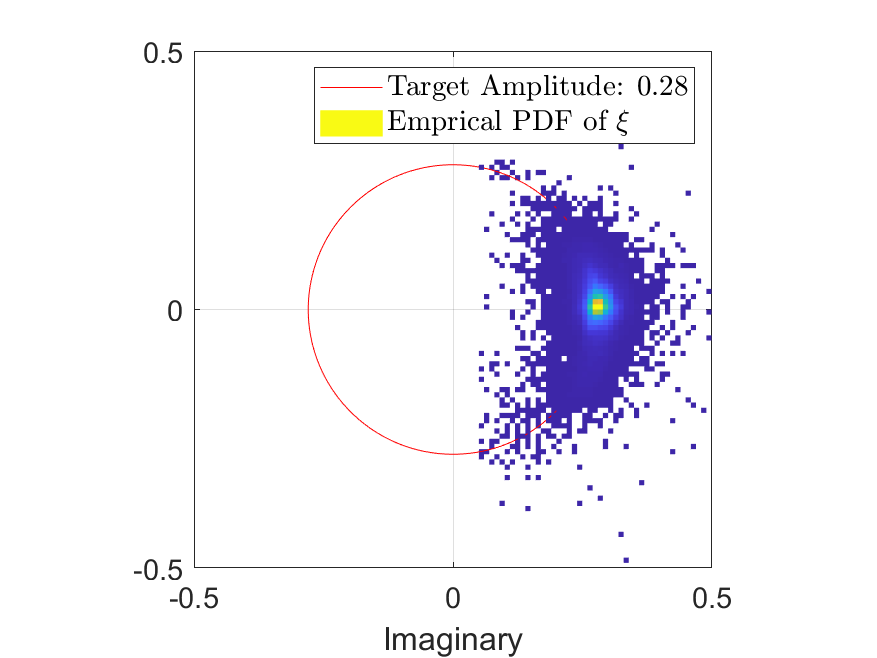}
	\label{subfig:ontheplane}
	}
	\subfloat[Distribution of centralized $|\xi_k|$.]{
	\includegraphics[width = \figuresizeSS]{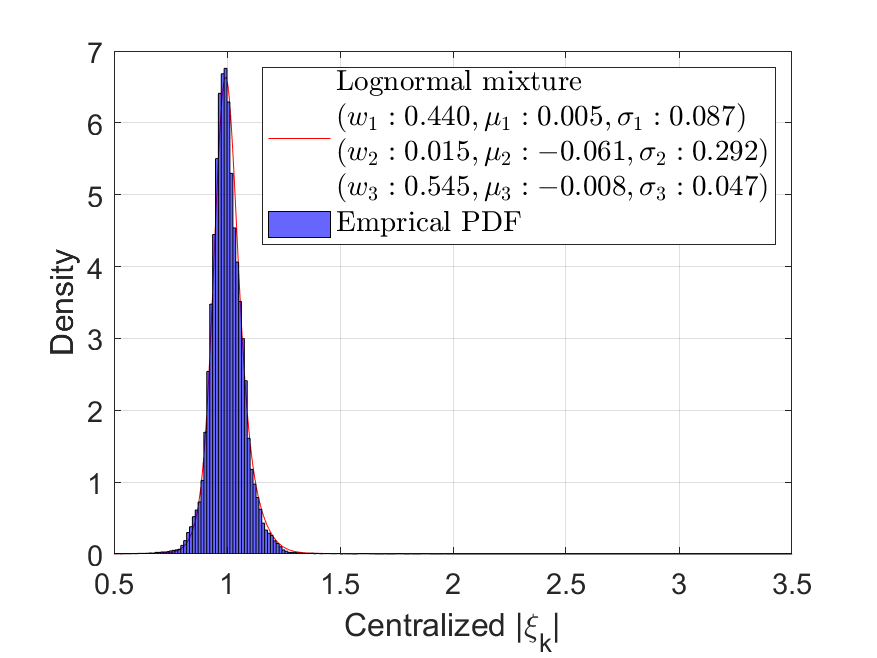}
	\label{subfig:abs}
	}
	\subfloat[Distribution of centralized $\angle{\xi_k}$.]{
	\includegraphics[width = \figuresizeSS]{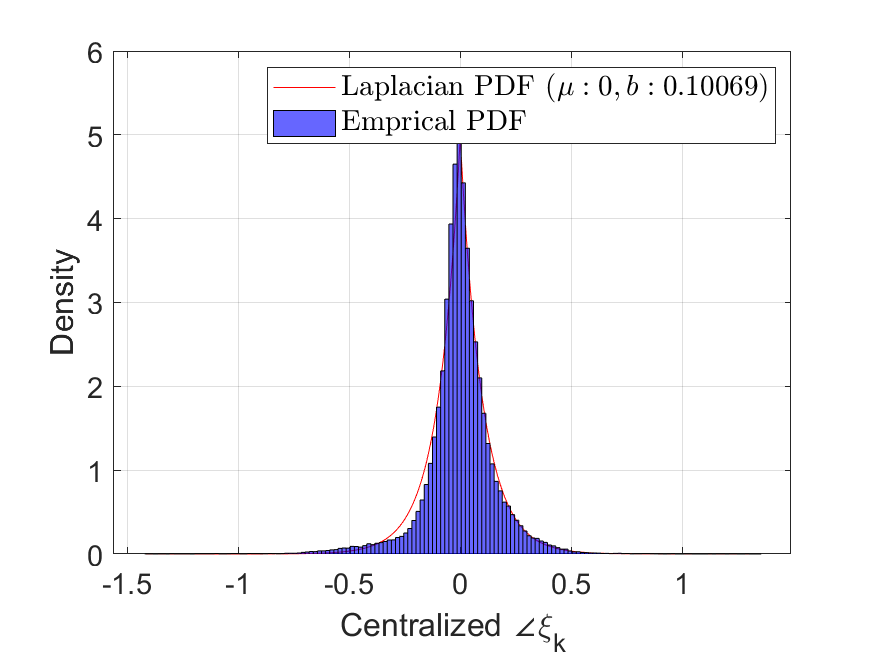}
	\label{subfig:angle}
	}	
	\subfloat[Measured CER in the testbed.]{
	\includegraphics[width = \figuresizeSS]{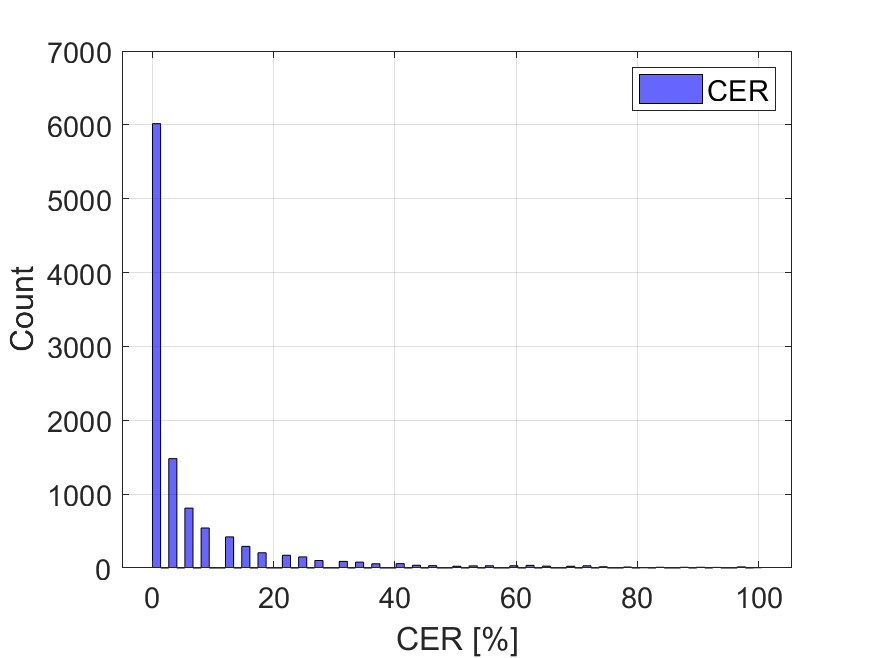}
	\label{subfig:cer}
}	

	\caption{Distribution of $\compositeModel[\indexED]$ based on the measurements for $\targetAlignment=0.28$ with zero radian phase and the fitted models for phase and amplitude distributions.}
	\label{fig:impairmentDis}
\end{figure*}
\subsection{Measurements}
\label{subsec:model}
During the measurements, we use the  protocol shown in \figurename~\ref{fig:protocol}. The ES first performs a sounding procedure to measure the amplitude alignment in the network. If the absolute value of the CFR of an ED deviates 0.75 dB more than the target level (i.e., $0.2$), the ES provides feedback regarding the attenuation factors for the corresponding ED. It then transmits a calibration command for frequency synchronization and re-measures the amplitude alignment with another sounding procedure. If no more feedback is required for amplitude synchronization, it transmits an aggregation command. Each ED loads its OAC waveform along with a trigger for the PCP, a sounding signal, and a test waveform into the SDR’s \ac{RAM}, and then waits for a trigger signal from the ES. After triggering, the ES measures $\compositeModel[\indexED]$ for all EDs through the sounding signals transmitted before the OAC symbols. If the synchronization is valid, it calculates the CER and records the results. We consider a valid synchronization if $\Re\{\compositeModel[\indexED]\}\ge5\times\Im\{\compositeModel[\indexED]\}$ and $|{\compositeModel[\indexED]}|\ge0.05$, $\forall\indexED$, hold.

The empirical distribution of $\compositeModel[\indexED]$ and the fitted models are given in \figurename~\ref{fig:impairmentDis} for 10000 measurements. In \figurename~\ref{fig:impairmentDis}\subref{subfig:ontheplane}, we show the empirical \ac{PDF} of the measured $\compositeModel[\indexED]$ with a clear scattering around the target value of $0.28$. In \figurename~\ref{fig:impairmentDis}\subref{subfig:abs} and \figurename~\ref{fig:impairmentDis}\subref{subfig:angle}, we show the distribution of the centralized amplitude, i.e., $|\compositeModel[\indexED]|/\arithmeticMean[{|\compositeModel[\indexED]|}]$, and phase, i.e., $\angle\compositeModel[\indexED]-\arithmeticMean[{\angle\compositeModel[\indexED]}]$, along with the fitted models. 
In this work, we model the centralized amplitude  as a mixture of three log-normal random variables, where the corresponding mean and standard deviations are given by $(0.005,-0.061,-0.008)$ and $(0.087,0.292,0.047)$, respectively, with the weights $(0.440,0.015,0.545)$, as in \figurename~\ref{fig:impairmentDis}\subref{subfig:abs}. For the centralized phase, we consider a Laplacian distribution with the mean $0$ and the scale $0.10069$, as shown in \figurename~\ref{fig:impairmentDis}\subref{subfig:angle}. %We use these models in our numerical results to estimate the \ac{CER} in practice.%, where we inherently assume that the angle and amplitude of  $\compositeModel[\indexED]$ are uncorrelated with each other. 

Finally, we provide the histogram of the measured \ac{CER} for the sum operation with $\numberOfDigits=1$, $\numberOfEdgeDevices=10$, and $\OACsymbolSize=8$ in \figurename~\ref{fig:impairmentDis}\subref{subfig:cer}. While the measured CER is $0\%$ in the majority of the cases, the computation can still fail due to the impairments.

\subsection{Enhanced Decoder}
Based on our observations in the testbed, it is likely that several nodes will be out-of-synch for a large $K$. Since even a single out-of-synch node can deteriorate the \ac{CER}, we now question if it is possible to improve the \ac{CER} via a better decoder when the channels, i.e., $\{\compositeModel[\indexED]\}$, are available at the fusion node. At the expense of increased complexity, one way of reducing \ac{CER} under impairments is a joint histogram and multi-user detection over the OAC symbols. To this end, we re-express \eqref{eq:superpositionPractice} as
\begin{align}
	\receivedSequence =\underbrace{\sum_{\indexED\notin\setOfNodesForDetector} \compositeModel[\indexED]\OACsymbol[{\symbol[\indexED]}]}_{{\text{in-synch}}}+ \underbrace{\sum_{\indexED\in\setOfNodesForDetector} \compositeModel[\indexED]\OACsymbol[{\symbol[\indexED]}]}_{\text{out-of-synch}}+\noiseVector~,
	\label{eq:superpositionPracticeDecoder}
\end{align}
where $\setOfNodesForDetector$ is a set of nodes that are considered as out-of-synch by the fusion node. Let  $\numberOfNodesForDetector$ denote the cardinality of $\setOfNodesForDetector$. Heuristically, we assume that the fusion node identifies $\setOfNodesForDetector$ by sorting  $|\compositeModel[\indexED]-\targetAlignment|$ and choosing $\numberOfNodesForDetector$ largest ones. By exploiting the function symmetry, we detect the histogram for $\numberOfEdgeDevices-\numberOfNodesForDetector$ nodes and the  symbols for $\numberOfNodesForDetector$ nodes jointly as
\begin{align}
	&\{\histogramDetected[],\symbolDetected[\indexED,\forall\indexED\in\setOfNodesForDetector] \}=\arg\min_{\substack{\histogram[]\in\setOfHistograms[\numberOfEdgeDevices-\numberOfNodesForDetector][\numberOfParameters]\\\symbol[\indexED,\forall\indexED\in\setOfNodesForDetector]}}\norm*{\receivedSequence - \targetAlignment\OACsymbolsMatrixOpt\histogram[] - \sum_{\indexED\in\setOfNodesForDetector} \compositeModel[\indexED]\OACsymbol[{\symbol[\indexED]}]}_2~. \nonumber
\end{align}
Thus, by restoring the histogram for $\numberOfEdgeDevices$ nodes, the decoder in \eqref{eq:MLdetector} can be re-expressed as
\begin{align}
	\decoder(\receivedSequence)=\functionHistogramArbitrary[{\histogramDetected[] + \sum_{\indexED\in\setOfNodesForDetector} \oneHotVector[{\symbolDetected[\indexED]}]}]~,
	\label{eq:MLdetectorImpairment}
\end{align}
where $\oneHotVector[i]$ denotes the vector of length $\numberOfParameters$ with a $1$ in the $i$th coordinate and $0$'s elsewhere.

The number of combinations for the histogram and  OAC symbols for $\numberOfNodesForDetector$ nodes can be calculated as  $\numberOfParameters^\numberOfNodesForDetector\times\binom{\numberOfEdgeDevices-\numberOfNodesForDetector+\numberOfParameters-1}{\numberOfParameters-1}$. Clearly, a larger $\numberOfNodesForDetector$ is not desirable as the complexity grows exponentially with $\numberOfNodesForDetector$.

\section{Numerical Results}
\label{sec:numerical}

\begin{figure}[t]
	\centering
	\includegraphics[width = \figuresize]{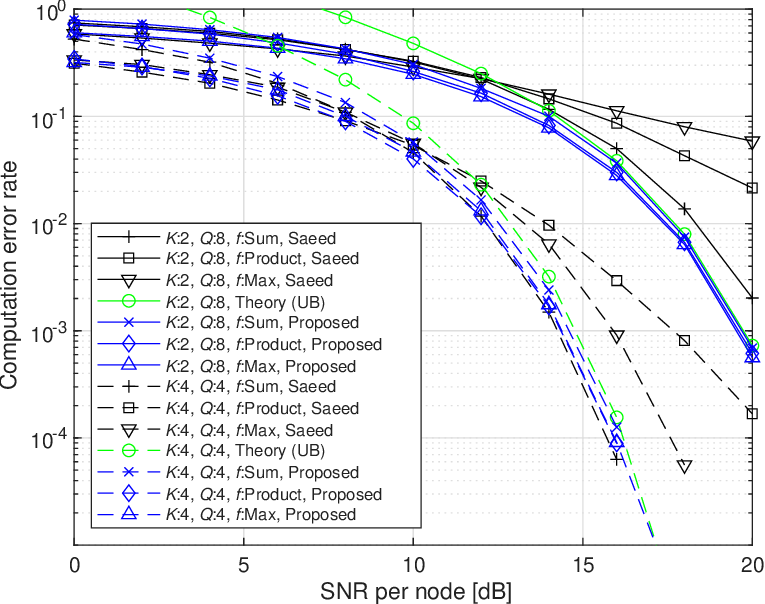}
	\caption{CER comparison ($\OACsymbolSize=1$).}
	\label{fig:CERperformanceComp}
\end{figure}

In this section, we numerically analyze the proposed construction,  considering both ideal synchronization and the impairment model introduced in Section~\ref{subsec:model}. We consider $\numberOfEdgeDevices\in\{2,4,10,20\}$ nodes, $\numberOfParameters=\{2,4,8\}$ parameters, $\OACsymbolSize\in\{1,2,4,8\}$ dimensions, and $\numberOfDigits\in\{1,2,4,8\}$ for  $\OACsymbolSize\numberOfDigits = \numberOfParameters$.  To demonstrate the usability of the proposed OAC symbols for an arbitrary symmetric function, we numerically calculate the \ac{CER} for  four different target functions, sum (i.e.,  $\sum_{\indexED=1}^{\numberOfEdgeDevices}\functionArgument[\indexED]$), product (i.e.,  $\prod_{\indexED=1}^{\numberOfEdgeDevices}\functionArgument[\indexED]$), maximum (i.e.,  $\max\{\functionArgument[1],\mydots,\functionArgument[\numberOfEdgeDevices]\}$), threshold function (i.e.,  $\indicatorFunction[{\sum_{\indexED=1}^{\numberOfEdgeDevices}\functionArgument[\indexED]}>\numberOfEdgeDevices(\numberOfParameters-1)/2]$) for $\functionArgument[\indexED]\in\{0,1,\mydots\numberOfParameters-1\}$ by using the \textit{same} constellation for a given $\OACsymbolSize$, $\numberOfDigits$, and $\numberOfEdgeDevices$. Without loss of generality, we set $\targetAlignment$ to be $1$ and sweep the \ac{SNR} per node (i.e., $\targetAlignment^2/\noiseVariance$) to obtain the \ac{CER} through Monte-Carlo simulations. For comparison, we use the OAC constellation optimized for specific functions (i.e., sum, product, and maximum) in \cite{Saeed_2024channelComp} and normalize them without centralization. For our results, we also plot the derived union bound in \eqref{eq:unionBoundCER} to justify our  results.

\def\spacingSub{2mm}
\begin{figure}
	\centering
	\subfloat[$\numberOfEdgeDevices=4$.]{\includegraphics[width = \figuresize]{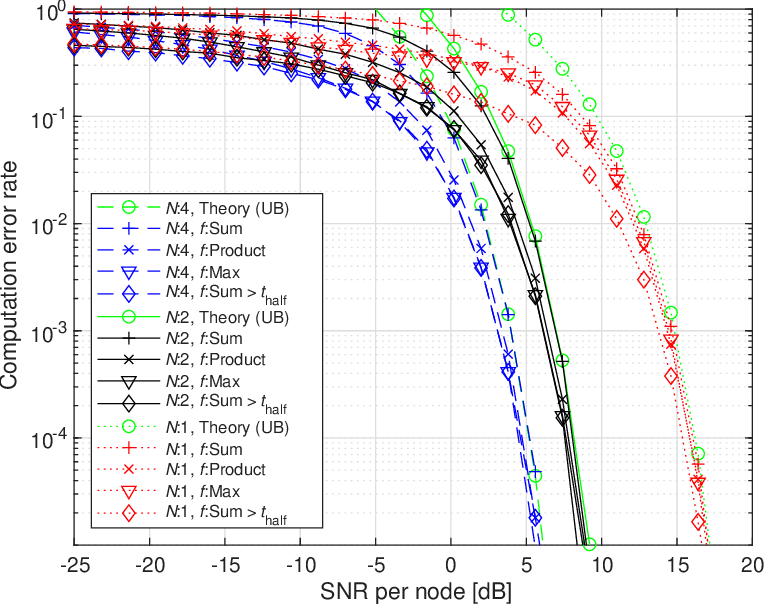}\label{subfig:CER_Q4K4}}\vspace{-\spacingSub}\\
	\subfloat[$\numberOfEdgeDevices=10$.]{\includegraphics[width = \figuresize]{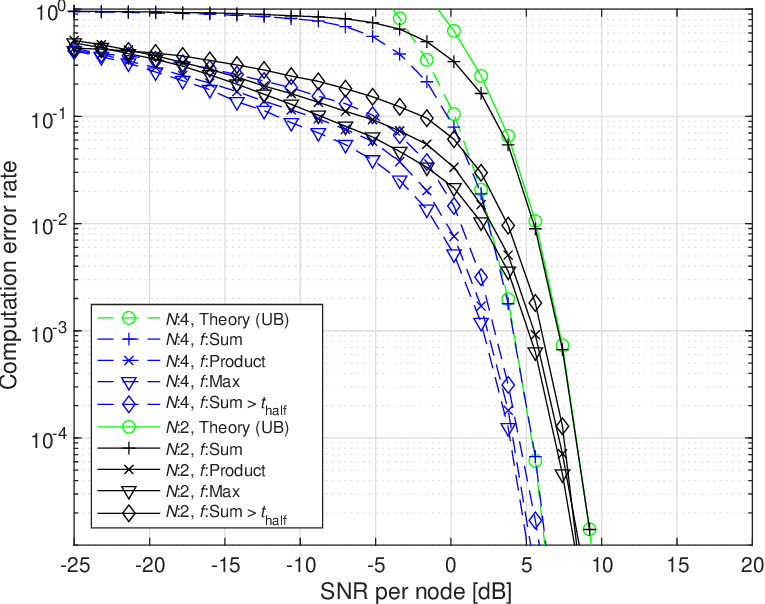}\label{subfig:CER_Q4K10}}\vspace{-\spacingSub}\\
	\subfloat[$\numberOfEdgeDevices=20$.]{\includegraphics[width = \figuresize]{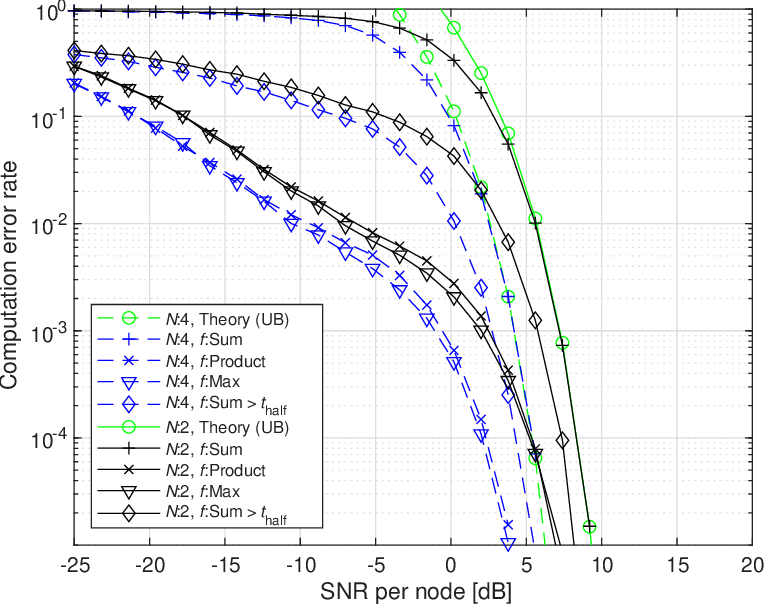}\label{subfig:CER_Q4K20}}\\	
	\caption{The CER performance for $\numberOfParameters=4$.}
	\label{fig:CERperformanceFixedQ4}
\end{figure}
\begin{figure}
	\centering
	\subfloat[$\numberOfEdgeDevices=4$.]{\includegraphics[width = \figuresize]{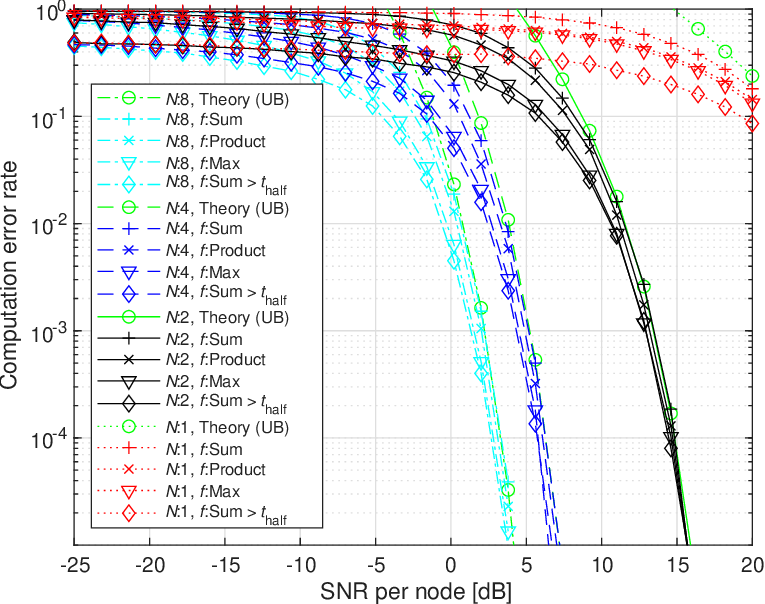}\label{subfig:CER_Q8K4}}\vspace{-\spacingSub}\\
	\subfloat[$\numberOfEdgeDevices=10$.]{\includegraphics[width = \figuresize]{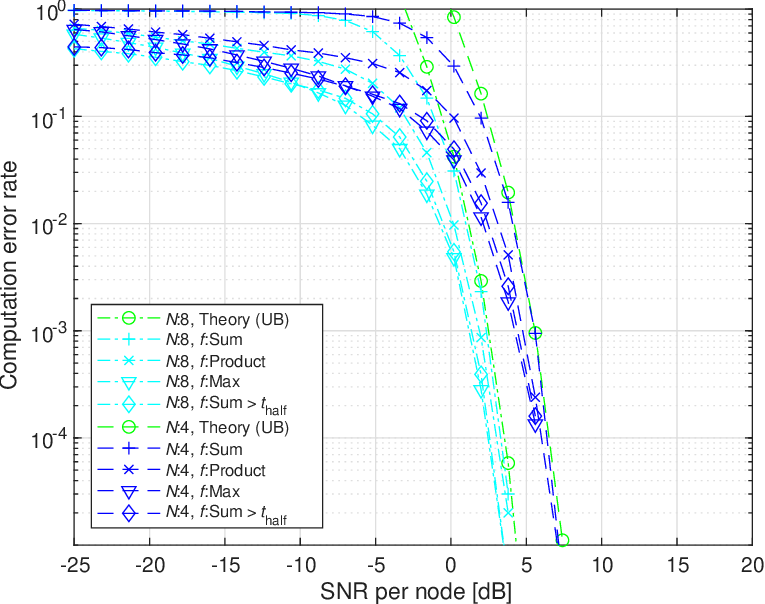}\label{subfig:CER_Q8K10}}\vspace{-\spacingSub}\\
	\subfloat[$\numberOfEdgeDevices=20$.]{\includegraphics[width = \figuresize]{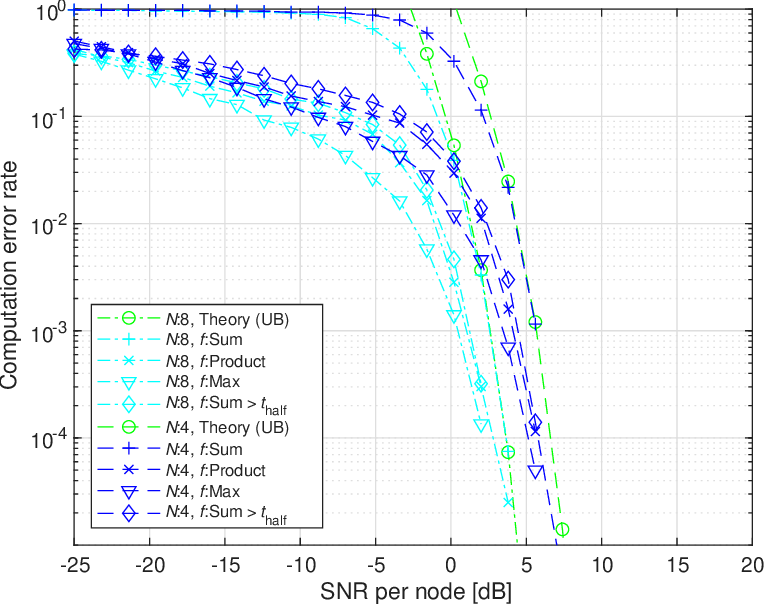}\label{subfig:CER_Q8K20}}\\	
	\caption{The CER performance for $\numberOfParameters=8$.}
	\label{fig:CERperformanceFixedQ8}
\end{figure}
In \figurename~\ref{fig:CERperformanceComp}, we consider the configurations used in \cite{Saeed_2024channelComp}, i.e., $(\numberOfEdgeDevices,\numberOfParameters)=(2,8)$ and  $(\numberOfEdgeDevices,\numberOfParameters)=(4,4)$ for $\OACsymbolSize=1$, implying that $\numberOfDigits=\numberOfParameters$  must hold for the proposed construction.  As shown in \figurename~\ref{fig:CERperformanceComp},  although the proposed construction uses the same constellations (i.e., \figurename~\ref{fig:constellationK2}\subref{subfig:K2constellation_n1d8} for $\numberOfEdgeDevices=2$ and  \figurename~\ref{fig:constellationK4}\subref{subfig:K4constellation_n1d4} for $\numberOfEdgeDevices=4$) for different target functions, it exhibits a superior performance for $\numberOfEdgeDevices=2$ and $\numberOfEdgeDevices=4$, except an approximately $0.6$~dB degradation for the sum function for $K=4$. Such degradation is indeed expected as \cite{Saeed_2024channelComp} aims at the constellation optimization for a given function and can leverage the histograms mapped to the same category. For example, the optimization  in  \cite{Saeed_2024channelComp}  leads to  a \ac{PAM} constellation and exploits such mappings for the sum function.\footnote{For $K=2$, the proposed construction performs better than the constellation in [6] for the sum function, as our construction uses the complex plane.} Nonetheless, the proposed construction offers a practical advantage by allowing one to use the same constellation for different symmetric functions depending on the application.  Finally, we note that the derived upper bound in \eqref{eq:unionBoundCER} is well-aligned with the numerical results.

In \figurename~\ref{fig:CERperformanceFixedQ4} and \figurename~\ref{fig:CERperformanceFixedQ8}, we compare the performance of the proposed multi-dimensional OAC constellations for $\numberOfParameters=4$ and  $\numberOfParameters=8$, respectively, and consider a larger $\numberOfEdgeDevices$. We observe that increasing $\OACsymbolSize$ offers a significant improvement in \ac{SNR} in all cases. 
For example, for $\numberOfEdgeDevices=4$ nodes,  to obtain  a \ac{CER}  of $10^{-4}$, the required \acp{SNR}  are approximately $16$~dB, $8$~dB, and $5$~dB for $\OACsymbolSize=1$, $\OACsymbolSize=2$, and $\OACsymbolSize=4$, respectively, as  shown in  \figurename~\ref{fig:CERperformanceFixedQ4}\subref{subfig:CER_Q4K4}.
When $\numberOfParameters$ is doubled, i.e., the case in \figurename~\ref{fig:CERperformanceFixedQ8}\subref{subfig:CER_Q8K4}, the same \ac{CER} requires approximately  $15$~dB, $6$~dB and $3$~dB \acp{SNR} for $\OACsymbolSize=2$, $\OACsymbolSize=4$, and $\OACsymbolSize=8$, respectively. As emphasized in Section~\ref{sec:proposedConstruction}, the reliability for  $\OACsymbolSize=1$ in  \figurename~\ref{fig:CERperformanceFixedQ8}\subref{subfig:CER_Q8K4} is severely limited, which indicates that increasing $\numberOfDigits$ for a large $\numberOfEdgeDevices$ may not be a viable solution in practice. 
For this reason, we conduct our experiments  with a larger number of nodes for  $\numberOfDigits=1$ and $\numberOfDigits=2$. 
As can be seen from  \figurename~\ref{fig:CERperformanceFixedQ4}\subref{subfig:CER_Q4K10} and  \figurename~\ref{fig:CERperformanceFixedQ4}\subref{subfig:CER_Q4K20}, the proposed construction achieves a \ac{CER} of $10^{-4}$ at $8$~dB and $5.3$~dB \acp{SNR} for $\numberOfDigits=2$ and $\numberOfDigits=1$ (i.e., sum function), respectively, for both $\numberOfEdgeDevices=10$ and $\numberOfEdgeDevices=20$.  \figurename~\ref{fig:CERperformanceFixedQ8}\subref{subfig:CER_Q4K10} and  \figurename~\ref{fig:CERperformanceFixedQ8}\subref{subfig:CER_Q8K20} show a similar relative result, where the required \acp{SNR} drop to $6$~dB and $3.5$~dB for $\numberOfDigits=2$ and $\numberOfDigits=1$, respectively, for the same CER. In \figurename~\ref{fig:CERperformanceFixedQ4}\subref{subfig:CER_Q4K20} and  \figurename~\ref{fig:CERperformanceFixedQ8}\subref{subfig:CER_Q4K20}, we  observe that some of the target functions are resilient against noise. This is because the cardinalities of the corresponding category spaces for the target functions, like maximum or threshold, are considerably less than $\maxMultiplicty$. Hence, even if the detected histogram is incorrect, the detector maps the incorrect histogram to a correct category for these target functions, leading to a small \ac{CER}.

\begin{figure}
	\centering
	\subfloat[$\numberOfParameters=4$.]{\includegraphics[width = \figuresize]{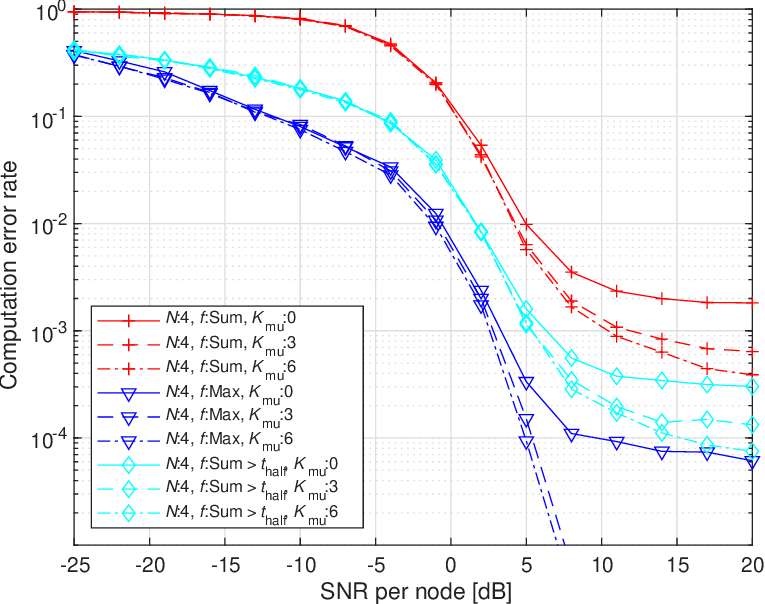}\label{subfig:CERimpairment_Q4K10}}\\
	\subfloat[$\numberOfParameters=8$.]{\includegraphics[width = \figuresize]{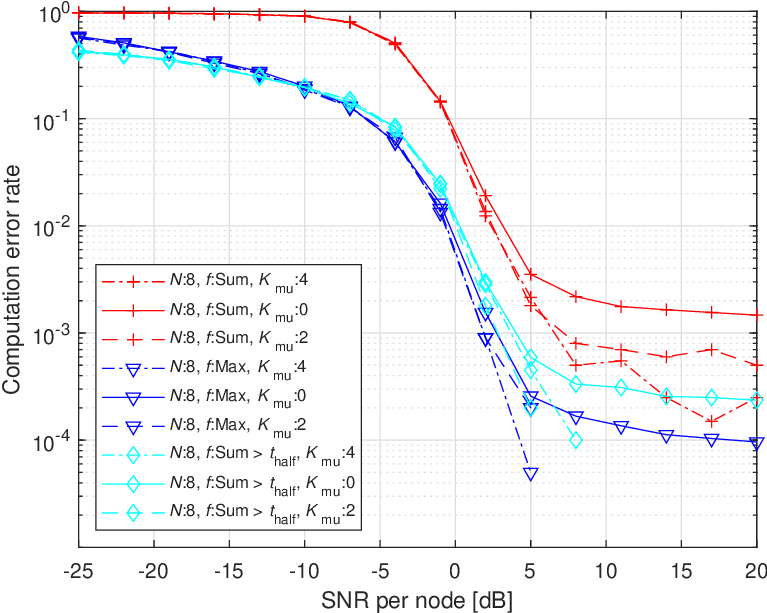}\label{subfig:CERimpairment_Q8K10}}\\
	\caption{The CER performance for $\numberOfEdgeDevices=10$.}
	\label{fig:CERperformanceImpairmentK10}
\end{figure}
Finally, in \figurename~\ref{fig:CERperformanceImpairmentK10}, we analyze the proposed scheme under the impairment model discussed in Section~\ref{subsec:model} for the enhanced decoder given in \eqref{eq:MLdetectorImpairment}. For this analysis, we consider $\numberOfEdgeDevices=10$ nodes and $\numberOfParameters\in\{4,8\}$ parameters for $\numberOfDigits=1$. We consider sum, maximum, and threshold functions, and evaluate the CER performance in \figurename~\ref{fig:CERperformanceImpairmentK10}\subref{subfig:CERimpairment_Q4K10} and \figurename~\ref{fig:CERperformanceImpairmentK10}\subref{subfig:CERimpairment_Q8K10} for $\numberOfParameters=4$ and $\numberOfParameters=8$, respectively, for a given $\numberOfNodesForDetector$. Compared with the corresponding results under ideal synchronization (i.e., \figurename~\ref{fig:CERperformanceFixedQ4}\subref{subfig:CER_Q4K10} and \figurename~\ref{fig:CERperformanceFixedQ8}\subref{subfig:CER_Q8K10}), we observe that synchronization impairment causes an error floor that depends on the target function, the number of parameters, and $\numberOfNodesForDetector$. For example, the error floors are about $2\times{10}^{-3}$, $3\times{10}^{-4}$, and $4\times10^{-5}$ for the sum, maximum, and threshold functions, respectively, for  $\numberOfParameters=4$ and $\numberOfNodesForDetector=0$, i.e., corresponding to the decoder in \eqref{eq:MLdetector}. By increasing $\numberOfNodesForDetector$, the error floor drops considerably at the expense of receiver complexity. For example, the error floor drops to  $4\times{10}^{-4}$ at $20$ dB \ac{SNR} for the sum function for $\numberOfNodesForDetector=6$, and the error floor is negligibly low for the maximum function in \figurename~\ref{fig:CERperformanceImpairmentK10}\subref{subfig:CER_Q4K10}. We observe similar trends  for $\numberOfParameters=8$, as shown in \figurename~\ref{fig:CERperformanceImpairmentK10}\subref{subfig:CER_Q8K10}.

\section{Concluding Remarks}
In this work, we propose a multi-dimensional symbol construction for digital OAC to compute \textit{any} symmetric function with a \textit{single} OAC constellation. After developing a theoretical framework for enumerating symmetric functions via their categorical representations and exploiting the sufficiency of the histogram of symbols for evaluating a symmetric function, inspired by TBMA, we construct multi-dimensional symbols that allow the detector to identify histograms for \textit{digital} function computation. Through both numerical results and theoretical assessments based on the union bound, we show that the proposed construction can provide a low \ac{CER} using the same OAC symbols for different functions, which can be beneficial for addressing diverse applications with a single design.

It is worth noting that the proposed construction is relevant to \ac{TBMA} as it uses histograms for computation. However, it differs from traditional \ac{TBMA} in that it uses histograms for digital function computation, rather than continuous statistical measures, and does not rely on orthogonal signaling. Furthermore, it allows one to pack more information in the multi-dimensional space via $\numberOfDigits$. However, our numerical results show that packing more information by using a large $\numberOfDigits$ for improved resource efficiency can degrade the \ac{CER} and limit the number of nodes in practice. In some cases (e.g., $\numberOfDigits\in\{1,2\}$), we show that the minimum distance is independent of the number of nodes, and that the corresponding \ac{CER} can be significantly reduced. For example,  a \ac{CER} of $10^{-4}$ can be achieved at approximately $3.5$ dB for $20$ nodes with $8$ parameters in an $8$-dimensional complex space. Also, the proposed approach obtains the corresponding multi-dimensional OAC symbols constructively for these cases.

The second major contribution of this work is the impairment model obtained based on realistic measurements, and the corresponding testbed. The testbed establishes a flexible framework for realistic measurements of methods that require synchronization in the UL direction, e.g., distributed antenna systems. Since it is based on low-cost off-the-shelf SDRs and host computers, it can also be replicated and expanded. It maintains time, frequency, phase, and amplitude synchronization in the network without using any auxiliary synchronization methods (e.g., \ac{GPS} or cables) to enable plausible measurements. For time and phase synchronizations, we implement a flexible trigger method mechanism in the FPGA of the SDRs and use the \ac{PCP} strategy. For procedures and coordination, such as \ac{UL} power control, we use a custom IEEE 802.11-like \ac{OFDM}-based \ac{PPDU}. In this work, we specifically model the amplitude and phase distortions on the coherent superposition by measuring the composite channel response at each node during the aggregation, which can be useful to assess the OAC schemes in practice. We then apply the obtained impairment model to the proposed multi-dimensional OAC symbols. Our results show that the impairments, unfortunately, cause an error floor in \ac{CER}. To address the error floor, we propose a joint multi-user and histogram detection and show that it can reduce the CER at the expense of receiver complexity, indicating a need for more research on low-complexity receivers.

Our efforts to realize coherent OAC in practice also shed light on several issues, opening new directions for future research. One major issue is that the probability of maintaining synchronization of \textit{all} nodes decreases with the number of nodes. Hence, there is still a need for an OAC scheme resilient against imperfection, potentially leveraging \textit{redundancy across the nodes} through coded function computation. One practical issue is the dynamic range of the superposed signal. Due to coherent aggregation, the superposed OAC signal can exceed the dynamic range of the \ac{ADC}, particularly with a large number of nodes. To circumvent this issue, one solution is to reduce the transmit power of each node, which, however, penalizes the \ac{SNR} per node and results in a higher \ac{CER}. On the contrary, if the receiver gain of the fusion node is reduced, it penalizes the strength of signals on orthogonal resources, e.g., the sounding signal. Addressing the dynamic range via new sampling strategies, e.g., modulo sampling \cite{Mulleti_2024}, or a low-PAPR design \textit{under} superposition are several interest angles that can be studied. We finally emphasize that amplitude synchronization is as challenging as phase synchronization and may be compromised by imperfections.

\bibliographystyle{IEEEtran}
\bibliography{references}

\end{document}